\documentclass[aps,pra,twocolumn,10pt,notitlepage,showpacs]{revtex4-2}

\usepackage{tikz}
\usetikzlibrary{calc,patterns}
\usepackage{ifthen}
\usepackage{amssymb}
\usepackage{mathtools}
\usepackage{dsfont}
\usepackage{bbm}
\usepackage{mleftright}\mleftright
\usepackage{thmtools} 
\usepackage{thm-restate}

\usepackage{algorithm2e}
\usepackage{enumitem}
\usepackage{float}

\usepackage{caption}
\usepackage{subcaption}

\newcommand{\eps}{\varepsilon}

\newcommand{\nrm}[1]{\left\lVert #1 \right\rVert}

\newcommand{\bigOt}[1]{\widetilde{\mathcal{O}}\left( #1 \right)}

\newcommand{\abs}[1]{\left|#1\right|}

\newcommand{\pvp}{\vec{p}{\kern 0.45mm}'}
\let\oldnabla\nabla
\renewcommand{\nabla}{\oldnabla\!}

\DeclarePairedDelimiter\bra{\langle}{\rvert}
\DeclarePairedDelimiter\ket{\lvert}{\rangle}
\DeclarePairedDelimiterX\braket[2]{\langle}{\rangle}{#1 \delimsize\vert #2}
\newcommand{\underflow}[2]{\underset{\kern-60mm \overbrace{#1} \kern-60mm}{#2}}

\providecommand{\spnorm}[1]{\left\lVert#1\right\rVert}

\long\def\ignore#1{}

\newtheorem{theorem}{Theorem}
\newtheorem{corollary}[theorem]{Corollary}
\newtheorem{lemma}[theorem]{Lemma}
\newtheorem{fact}[theorem]{Fact}
\newtheorem{prop}[theorem]{Proposition}

\newcommand{\A}{\ensuremath{\mathcal{A}}}

\newcommand{\C}{\ensuremath{\mathcal{C}}}

\newcommand{\Oo}{\ensuremath{\mathcal{O}}}

\newcommand{\I}{\ensuremath{\mathcal{I}}}

\newcommand{\U}{\ensuremath{\mathcal{U}}}

\newcommand{\M}{\mathcal{M}}
\usepackage{tikz}

\usetikzlibrary{backgrounds,fit,decorations.pathreplacing,calc}

\usepackage{qcircuit}

\newenvironment{proof}
{\noindent {\bf Proof. }}
{{\hfill $\Box$}\\	\smallskip}

\usepackage[final]{hyperref}
\hypersetup{
	colorlinks = true,
	allcolors = {blue},
}

\usepackage{comment}
\usepackage{amsmath}
\usepackage{color}
\usepackage{graphicx}
\usepackage{subcaption}
\usepackage{caption}
\usepackage{amssymb}
\usepackage{amsfonts}
\usepackage[parfill]{parskip}
\usepackage{hyperref}

\usepackage[format=plain,justification=centerlast,singlelinecheck=true]{caption}

\begin{document}


\title{Analog quantum algorithms for the mixing of Markov chains}


\author{Shantanav Chakraborty$^1$}  
\email[]{shchakra@ulb.ac.be}
\author{Kyle Luh$^2$}
\email[]{kluh@cmsa.fas.harvard.edu}
\author{J\'{e}r\'{e}mie Roland$^1$}
\email[]{jroland@ulb.ac.be}
\affiliation{$^1$QuIC, Ecole Polytechnique de Bruxelles, Universit\'{e} libre de Bruxelles}
\affiliation{$^2$Center for Mathematical Sciences and Applications, Harvard University}
\begin{abstract}
The problem of sampling from the stationary distribution of a Markov chain finds widespread applications in a variety of fields. The time required for a Markov chain to converge to its stationary distribution is known as the classical mixing time. In this article, we deal with analog quantum algorithms for mixing. First, we provide an analog quantum algorithm that given a Markov chain, allows us to sample from its stationary distribution in a time that scales as the sum of the square root of the classical mixing time and the square root of the classical hitting time. Our algorithm makes use of the framework of interpolated quantum walks and relies on Hamiltonian evolution in conjunction with von Neumann measurements. 

There also exists a different notion for quantum mixing: the problem of sampling from the limiting distribution of quantum walks, defined in a time-averaged sense. In this scenario, the quantum mixing time is defined as the time required to sample from a distribution that is close to this limiting distribution. 
Recently we provided an upper bound on the quantum mixing time for Erd\"os-Renyi random graphs [Phys. Rev. Lett. 124, 050501 (2020)]. Here, we also extend and expand upon our findings therein. Namely, we provide an intuitive understanding of the state-of-the-art random matrix theory tools used to derive our results. In particular, for our analysis we require information about macroscopic, mesoscopic and microscopic statistics of eigenvalues of random matrices which we highlight here. Furthermore, we provide numerical simulations that corroborate our analytical findings and extend this notion of mixing from simple graphs to any ergodic, reversible, Markov chain.
\end{abstract}
\date{\today}
\maketitle
\section{Introduction}
Markov chain-based algorithms are applied in a plethora of fields ranging from statistical physics \cite{binder1993monte}, combinatorial optimization \cite{aarts1985statistical} to network science \cite{page1999pagerank} and form the basis of Markov chain Monte Carlo-based methods \cite{gilks1995markov}. In many of these applications, the underlying task is often to sample from the so-called steady state (also known as the stationary distribution) of the associated Markov chain. 

One way to sample from a stationary distribution is by \textit{mixing}. The Markov chain, which is represented by a stochastic matrix $P$ is applied repeatedly to some initial distribution. The resultant random walk reaches a final distribution that is close to a stationary distribution of $P$, irrespective of the initial distribution. For most applications, the Markov chain is \textit{ergodic}, implying that it has a unique stationary distribution and, \textit{reversible}, i.e.\ it satisfies detailed balance. (We refer the reader to Sec.~\ref{sec:preliminaries} for details on the definitions of these terms related to Markov chains). Henceforth, unless stated otherwise, we shall restrict our attention to ergodic, reversible Markov chains. For a given Markov chain $P$, the minimum time after which the distribution is $\epsilon$-close to the stationary distribution is known as the \textit{mixing time} of the random walk on $P$. It is well known that the \textit{mixing time} is related to the spectral gap of $P$. For an ergodic, Markov chain with spectral gap $\Delta$, the mixing time is in $\bigOt{1/\Delta}$ \footnote{Throughout the article, we use $\tilde{\mathcal{O}}(.)$ and $\tilde{\Omega}(.)$ to suppress polylogarithmic factors, i.e.\ $\tilde{\mathcal{O}}(f(n))=\mathcal{O}(f(n)\mathrm{polylog}(f(n)))$ and $\tilde{\Omega}(f(n))=\Omega(f(n)\mathrm{polylog}(f(n)))$.}.

The stationary distribution, by definition, is the limiting distribution of the resultant random walk on $P$, i.e.\ once the stationary state is reached, the random walk ceases to evolve. This implies that as $t\rightarrow\infty$, $P^t$ applied to any initial distribution converges to the stationary distribution. Thus the classical mixing time is also the time required to sample from the limiting distribution of the underlying random walk.   

In the context of quantum algorithms, there arise two notions of mixing and hence of mixing time. First, it is natural to consider whether, given a Markov chain $P$, a quantum algorithm can allow us to prepare a coherent encoding of the stationary distribution of $P$. We shall refer to this problem as $QSSamp$. Measuring the output state of such an algorithm would enable us to sample from the (classical) stationary state of $P$. Preparing such a coherent encoding also has other applications which we discuss later. 

The other notion of mixing arises from considering the limiting distribution of the underlying quantum walk itself. As quantum evolutions are unitary and hence distance-preserving, there is no inherent limiting stationary distribution for quantum walks. However, it turns out that one can define a limiting distribution of the quantum walk in a time-averaged sense. 

Starting from some initial state, one can obtain the probability that the walker is in some final state after a time $t$ which is picked uniformly at random in the interval $[0,T]$. This gives a \textit{time-averaged probability distribution} at any time $T$ and also a \textit{limiting probability distribution} as $T\rightarrow\infty$. The problem of sampling from this time-averaged limiting distribution of a quantum walk gives rise to another notion of mixing and we shall refer to this problem as $QLSamp$. The \textit{mixing time} of a quantum walk is then defined as the time after which the time-averaged probability distribution is close to the limiting probability distribution, i.e.\ the time required to solve $QLSamp$.  

In this article, we deal with both $QSSamp$ and $QLSamp$ problems. We provide the first purely analog quantum algorithm to solve the $QSSamp$ problem while for the $QLSamp$ problem, we expand and extend upon the results of Ref.~\cite{chakraborty2020fast}, where we prove an upper bound for the quantum mixing time for \textit{almost all graphs}.

Aharonov and Ta-Shma \cite{aharonov2003adiabatic} demonstrated that the existence of an efficient quantum algorithm for $QSSamp$ would imply that problems in the complexity class Statistical Zero Knowledge (SZK) such as \textit{Graph Isomorphism} would be solvable in polynomial-time using a quantum computer (BQP), i.e.\ $\mathrm{SZK}\subseteq \mathrm{BQP}$. This would be a surprising result as such a generic $QSSamp$ algorithm would be oblivious to the specific structure of the underlying problem. For example, consider the problem of \textit{Graph Isomorphism} \cite{aho1974design} (deciding whether two graphs are isomorphic to each other). Given graphs $G_1$ and $G_2$, a quantum algorithm for mixing could be used to prepare states that are a uniform superposition of all graphs that are isomorphic to them. If these states are equal, then $G_1$ and $G_2$ are isomorphic. A simple SWAP Test could then be used in conjunction with a quantum algorithm for $QSSamp$  to solve \textit{Graph Isomorphism}. Thus, generic quantum algorithms for $QSSamp$ are unlikely to be efficient.
   
Having said that, there do exist quantum algorithms that solve this problem \cite{richter2007quantum, wocjan2008speedup, dunjko2015quantum}, some of which have even been instrumental in obtaining speedups for quantum machine learning \cite{paparo2014quantum, orsucci2018optimal, dunjko2018machine}. Richter \cite{richter2007almost} conjectured that one could construct a quantum algorithm for this problem that has a running time that is in $\widetilde{\mathcal{O}}(1/\sqrt{\Delta})$, yielding a quadratic speedup over its classical counterpart. Developing quantum algorithms that match this conjectured bound have been challenging. Most of the existing quantum algorithms are based on Szegedy's framework for discrete-time quantum walks \cite{szegedy2004quantum}. 

The key idea that encompasses all existing algorithms for $QSSamp$ is to make use of the so called quantum spatial search algorithm \cite{magniez2011search}. Given an ergodic, reversible Markov chain $P$ with a set of marked nodes, a spatial search algorithm finds an element from this marked set. Classically, this task requires a time known as the \textit{hitting time} of the corresponding random walk on $P$. It has been shown that a discrete-time quantum walk-based quantum algorithm for spatial search can accomplish this task quadratically faster (up to logarithmic factors) \cite{krovi2016quantum, ambainis2019quadratic}. Such quantum algorithms start from the coherent encoding of the stationary state of $P$ (it inherently assumes that this state can be prepared efficiently) and end up in a state that has a constant overlap with an element from the marked set. Thus, intuitively, quantum spatial search algorithms can be run in reverse to obtain quantum mixing algorithms. However, simply obtaining a constant overlap with the stationary state is not enough and these mixing algorithms require the use of quantum phase estimation \cite{cleve1998quantum} and quantum amplitude amplification \cite{brassard2002quantum} to solve the $QSSamp$ problem. Recently, Apers and Sarlette provided a quantum algorithm that can quadratically \textit{fast-forward} the dynamics of Markov chains which can also be used to solve the $QSSamp$ problem \cite{simon2018quantum}. The running time of these algorithms scale as the square root of the \textit{hitting time} of the corresponding quantum walk on the underlying Markov chain.  

To the best of our knowledge, there do not exist any analog quantum algorithm for solving the $QSSamp$ problem. In this framework, key algorithmic primitives such as quantum phase estimation and quantum amplitude amplification are missing as they are inherently discrete-time. In order to construct an analog quantum algorithm for $QSSamp$ we assume that, given an ergodic, reversible Markov chain $P$, we have access to a time-independent Hamiltonian that encodes the connectivity of $P$. This Hamiltonian, defined in Sec.~\ref{sec-main:hamiltonian-any-markov-chain}, corresponds to a quantum walk on the edges of $P$.  Furthermore it has been recently used to design continuous-time quantum walk-based quantum algorithms for spatial search that can find a single marked node on any ergodic, reversible Markov chain in square root of the hitting time \cite{chakraborty2018finding}. We use the time-evolution of this Hamiltonian as the key primitive to our algorithm. The second key primitive is to use von Neumann measurements \cite{von1955mathematical} for quantum state generation. Childs et al. used a sequence of such von Neumann measurements as an alternative to adiabatic quantum computation and for solving combinatorial search algorithms \cite{childs2002quantum}. In Sec.~\ref{sec-main:von-neumann} we demonstrate that this scheme can be used to prepare eigenstates of Hamiltonians.

We show (Sec.~\ref{sec-main:search-pointer}) that these two primitives allow us to develop a continuous-time quantum walk based algorithm for spatial search. This algorithm differs from the one developed in Ref.~\cite{chakraborty2018finding} which makes use of quantum phase randomization \cite{somma2010quantum}. It provides an alternative scheme by which one can find an element in a marked set of states of any ergodic, reversible Markov chain in square root of the \textit{extended hitting time}. Although this algorithm has the same running time as that of Ref.~\cite{chakraborty2018finding}, it provides useful intuition about how to build an analog quantum algorithm for $QSSamp$. 

Our quantum algorithm for mixing, explained in detail in Sec.~\ref{sec:preparation-of-stationary-state-von-neumann}, avoids the need for amplitude amplification by making use of the framework of interpolated Markov chains and switching between two different values of the interpolation parameter. The running time scales as the sum of the square root of the classical mixing time and the square root of the hitting time.

We also discuss the problem of $QLSamp$ on generic graphs. The limiting distribution of quantum walks can be quite different from that obtained from a quantum algorithm for solving $QSSamp$. Unlike its classical counterpart, for $QLSamp$, the limiting distribution is dependent on the initial state of the quantum walk. Moreover, instead of being dependent on the spectral gap $\Delta$, the quantum mixing time depends on all eigenvalue gaps of the underlying Hamiltonian. Aharonov et al.~\cite{aharonov2001quantum} were the first to study this problem. They showed that a discrete-time quantum walk on the cycle graph mixes faster than its classical counterpart. Since then several works have considered the mixing time of both continuous and discrete-time quantum walks on specific graphs  \cite{ahmadi2003mixing,kendon2003decoherence,fedichkin2005mixing,marquezino2008mixing,marquezino2010mixing,kieferova2012quantum}. The upper bound for the mixing time of quantum walks have been proven to be slower than its classical counterpart for some graphs while a quadratic speedup has been obtained for others.

Recently, we proved an upper bound for the mixing time of quantum walks for \textit{almost all graphs} \cite{chakraborty2020fast}. This implies that the fraction of graphs of $n$ nodes for which our upper bound holds, goes to 1 as $n$ goes to infinity or equivalently, if a graph is picked uniformly at random from the set of all graphs, our result provides an upper bound on the quantum mixing time \textit{almost surely}, i.e.\ with probability $1-o(1)$. Throughout the article, we shall use the phrase \textit{almost all graphs} to signify precisely this. 

We proved this by obtaining the mixing time for quantum walks on Erd\"os-Renyi random graphs: graphs of $n$ nodes such that the probability of an edge existing between any two nodes is $p$, typically denoted as $G(n,p)$. Here, we expand upon the results of the letter \cite{chakraborty2020fast}. In particular, our goal is to offer an intuitive explanation of our proof techniques with an emphasis on the several recently developed random matrix theory tools that were used to derive the aforementioned results. We also corroborate our analytical findings numerically and also extend the notion of $QLSamp$ to any ergodic, reversible Markov chain. In fact, our numerical findings confirm the fact that the mixing time for quantum walks on $G(n,p)$ is in $\bigOt{n^{3/2}}$ for dense random graphs (constant $p$). Additionally they also show that the limiting probability distribution is close to the uniform distribution. 

This article is organized as follows: In Sec.~\ref{sec:preliminaries}, we explain some basic concepts and quantities related to Markov chains that we shall use in subsequent sections. In Sec.~\ref{sec-main:von-neumann} we show how von Neumann measurements can be used for preparing eigenstates of Hamiltonians. In Sec.~\ref{sec-main:hamiltonian-any-markov-chain}, we define a Hamiltonian corresponding to a quantum walk on the edges of any ergodic, reversible Markov chain. In Sec.~\ref{sec-main:search-pointer}, we make use of von Neumann measurements and Hamiltonian evolution to provide a quantum algorithm for spatial search. This provides an intuitive understanding of our analog quantum algorithm for solving $QSSamp$, which we describe in Sec.~\ref{sec:preparation-of-stationary-state-von-neumann}. Next, in Sec.~\ref{sec-main:mixing-time-erg}, we deal with solving the $QLSamp$ problem. Finally, we conclude with a brief discussion and summary in Sec.~\ref{sec-main:discussion}.
\section{Preliminaries}
\label{sec:preliminaries}
In this section we state some basic definitions about Markov chains which we shall use subsequently.
\subsection{Basics of Markov chains} 
\label{subsec:basics-mc}
A Markov chain on a discrete state space $X$, such that $|X|=n$, can be described by a $n\times n$ stochastic matrix $P$ \cite{norris1998markov}. Each entry $p_{xy}$ of this matrix $P$ represents the probability of transitioning from state $x$ to state $y$. Any distribution over the state space of the Markov chain is represented by a stochastic row vector. 

A Markov chain is \textit{irreducible} if any state can be reached from any other state in a finite number of steps. Any \textit{irreducible} Markov chain is \textit{aperiodic} if there exists no integer greater than one that divides the length of every directed cycle of the graph. A Markov chain is \textit{ergodic} if it is both \textit{irreducible} and \textit{aperiodic}. By the Perron-Frobenius Theorem, any ergodic Markov chain $P$ has a unique stationary state $\pi$ such that $\pi P=\pi$. The stationary state $\pi$ is a stochastic row vector and has support on all the elements of $X$. Let us denote it as 
\begin{equation}
\label{eqmain:stationary-state-classical}
\pi=\left(\pi_1~~\pi_2~~\cdots~~\pi_n\right),
\end{equation}
such that $\sum_{j=1}^n \pi_j=1$. Starting from any initial probability distribution $\mu$ over the state space $X$, the repeated application of $P$ leads to convergence to the stationary distribution $\pi$, i.e.\ $\lim_{t\rightarrow\infty} \mu P^t=\pi$. This is known as the \textit{mixing} of a Markov chain. It follows from the Perron-Frobenius theorem that other than $\pi$, all eigenvectors have eigenvalues of absolute value strictly less than $1$. Thus, $\pi$ is the unique eigenvector with eigenvalue $1$ and all other eigenvalues lie between $-1$ and $1$. Throughout the paper we shall be working with the Markov chain corresponding to the \textit{lazy walk}, i.e.\ we shall map $P\mapsto (I+P)/2$. This transformation ensures that all the eigenvalues of $P$ lie between $0$ and $1$. This transformation will not affect our results other than by a factor of two, which is irrelevant in the asymptotic limit. Throughout the article, we shall denote the gap between the two highest eigenvalues of $P$ (the spectral gap) by $\Delta$.

Let $p_{x,y}$ denote the $(x,y)^{\text{th}}$-entry of the ergodic Markov chain $P$ with stationary state $\pi$. Then the $(x,y)^{\text{th}}$ entry of the time-reversed Markov chain of $P$, denoted by $P^*$, is 
$$p^*_{x,y}=p_{y,x}\dfrac{\pi_y}{\pi_x}.$$
We shall concern ourselves with ergodic Markov chains that are also \textit{reversible}, i.e. Markov chains for which \ $P=P^*$. Any reversible $P$ satisfies the \textit{detailed balance condition} 
$$\pi_x p_{xy}=\pi_y p_{yx},~\forall (x,y)\in X.$$ 
This can also be rewritten as 
$$\text{diag}(\pi)P=P^T\text{diag}(\pi),$$
where $\text{diag}(\pi)$ is a diagonal matrix with the $j^{\text{th}}$ diagonal entry being $\pi_j$. In other words, the reversibility criterion implies that the matrix $\text{diag}(\pi)P$ is symmetric. Henceforth we shall only deal with reversible (and hence ergodic) Markov chains.
~\\~\\
\textbf{Interpolated Markov chains:} Let us assume that a subset of the elements of the state space of the Markov chain $P$ is marked. Let $M\subset X$ denote the set of marked elements. Given any $P$, we define $P'$ as the \textit{absorbing Markov chain} obtained from $P$ by replacing all the outgoing edges from $M$ to $X$ by self-loops. If we re-arrange the elements of $X$ such that the unmarked elements $U:=X\backslash M$ appear first, then we can write 
\begin{align}
P=\begin{bmatrix}
P_{UU} & P_ {UM}\\
P_{MU} & P_{MM}
\end{bmatrix},~~~~~~~~~P'=\begin{bmatrix}
P_{UU} & P_ {UM}\\
0 & I
\end{bmatrix},
\end{align}
where $P_{UU}$ and $P_{MM}$ are square matrices of size $(n-|M|)\times (n-|M|)$ and $|M|\times |M|$ respectively. On the other hand $P_{UM}$ and $P_{MU}$ are matrices of size $(n-|M|)\times |M|$ and $|M|\times (n-|M|)$ respectively. Then the \textit{interpolated Markov chain} is defined as 
\begin{equation}
\label{eqmain:interpolated-mc-defintion}
P(s)=(1-s)P+sP',
\end{equation}
where $s\in[0,1]$. The interpolated Markov chain thus has a block structure
\begin{align}
P=\begin{bmatrix}
P_{UU} & P_ {UM}\\
(1-s)P_{MU} & (1-s)P_{MM}+sI
\end{bmatrix}.
\end{align}
Clearly, $P(0)=P$ and $P(1)=P'$. Notice that if $P$ is ergodic, so is $P(s)$ for $s\in [0,1)$. This is because any edge in $P$ is also an edge of $P(s)$ and so the properties of \textit{irreducibility} and \textit{aperiodicity} are preserved. However when $s=1$, $P(s)$ has outgoing edges from $M$ replaced by self-loops and as such the states in $U$ are not accessible from $M$, implying that $P(1)$ is not ergodic. The spectral gap of $P(s)$ is denoted by $\Delta(s)$. 

Now we shall see how the stationary state of $P$ is related to that of $P(s)$. Since $X=U \cup M$, the stationary state $\pi$ can be written as 
\begin{equation}
\label{eqmain:stationary-state-split}
\pi=(\pi_U~~\pi_M),
\end{equation}
where $\pi_U$ and $\pi_M$ are row-vectors of length $n-|M|$ and $|M|$ respectively. As mentioned previously, $P'$ is not ergodic and does not have a unique stationary state. In fact, any state having support over only the marked set is a stationary state of $P'$. 

On the other hand $P(s)$ is ergodic for $s\in [0,1)$. Let $p_M=\sum_{x\in M}\pi_x$ be the probability of obtaining a marked element in the stationary state of $P$. Then it is easy to verify that the unique stationary state of $P(s)$ is 
\begin{equation}
\label{eqmain:stationary-state-interpolated-mc-classical}
\pi(s)=\dfrac{1}{1-s(1-p_M)}\left((1-s)\pi_U~~\pi_M\right).
\end{equation}
~\\~\\
\textbf{Discriminant matrix:~} We denote by 
\begin{equation}
\label{eqmain:discriminant-matrix-definition}
D(P(s))=\sqrt{P(s)\circ P(s)^T}
\end{equation}
the symmetric matrix whose $(x,y)^{\mathrm{th}}$ entry is $D_{xy}(P(s))=\sqrt{p_{xy}(s)p_{yx}(s)}$.  Here $\circ$ indicates the Hadamard product. 
~\\
For any $s\in[0,1)$ as $P(s)$ is reversible, the detailed-balance condition is satisfied. So, each entry of $D(P(s))$ can be expressed as 
\begin{align}
D_{xy}(P(s))&=\sqrt{p_{xy}(s)p_{yx}(s)}\\
            &=p_{xy}(s)\sqrt{\frac{\pi_x(s)}{\pi_y(s)}}.
\end{align}
This leads us to the following fact:
~\\
\begin{fact}
\label{fact:reversibility-condition-discriminant-matrix}
For any ergodic, reversible Markov chain $P$, we have that for $s\in[0,1)$ 
$$D(P(s))=\mathrm{diag}(\sqrt{\pi(s)})P(s)\mathrm{diag}(\sqrt{\pi(s)})^{-1},$$
where $\sqrt{\pi(s)}$ is a row vector with its $j^{\text{th}}$-entry being $\sqrt{\pi_j(s)}$. 
\end{fact}
~\\
From Fact \ref{fact:reversibility-condition-discriminant-matrix}, it follows that $D(P(s))$ is \textit{similar} to $P(s)$, i.e.\ they have the same set of eigenvalues \footnote{Note that throughout this paper, we shall be dealing with values of $s\in[0,1)$ and so the properties of $D(P(1))$ are not relevant here.}.

Let the spectral decomposition of $D(P(s))$ be
\begin{equation}
\label{eqmain:discriminant-matrix-spectral}
D(P(s))=\sum_{i=1}^{n}\lambda_i(s)\ket{v_i(s)}\bra{v_i(s)},
\end{equation}
where $\ket{v_i(s)}$ is an eigenvector of $D(P(s))$ with eigenvalue $\lambda_i(s)$. Furthermore, $\lambda_n(s)=1 > \lambda_{n-1}(s)\geq\cdots\geq \lambda_1(s)$. 
~\\
\begin{fact}
\label{fact:highest-eigenstate-equals-pi-s}
For $s\in[0,1)$, the eigenstate of $D(P(s))$ with eigenvalue $1$ is given by
$$\ket{v_n(s)}=\sqrt{\pi(s)^T},$$
where $\sqrt{\pi(s)}$ is a row vector with its $j^{\text{th}}$-entry being $\sqrt{\pi_j(s)}$. 
\end{fact}
~\\
This fact follows from the reversibility condition stated in Fact~\ref{fact:reversibility-condition-discriminant-matrix}, i.e.\ for $s\in [0,1)$ we have
\begin{align}
D(P(s))\sqrt{\pi(s)^T}&=\mathrm{diag}(\sqrt{\pi(s)})P(s)\mathrm{diag}(\sqrt{\pi(s)})^{-1}\sqrt{\pi(s)^T}\\
                      &=\sqrt{\pi(s)^T}.
\end{align}
The $1$-eigenvector of $D(P(s))$,~$\ket{v_n(s)}$ can also be expressed in a different form. 
~\\
\begin{prop}
\label{prop-main:zero-estate-somma-ortiz}
The eigenstate of eigenvalue $1$ of $D(P(s))$ can be expressed as 
\begin{equation}
\label{eqmain:highest-eigenstate-discriminant-matrix}
\ket{v_n(s)}=\sqrt{\dfrac{(1-s)(1-p_M)}{1-s(1-p_M)}}\ket{U}+\sqrt{\dfrac{p_M}{1-s(1-p_M)}}\ket{M},
\end{equation} 
where $\ket{U}$ and $\ket{M}$ are defined as
\begin{align}
\ket{U}&=\frac{1}{\sqrt{1-p_M}}\sum_{x\notin M}\sqrt{\pi_x}\ket{x}\\
\ket{M}&=\frac{1}{\sqrt{p_M}}\sum_{x\in M}\sqrt{\pi_x}\ket{x}.
\end{align}  
\end{prop}
~\\
This follows directly from Fact \ref{fact:highest-eigenstate-equals-pi-s}. 
~\\~\\
%
\subsection{Some quantities related to Markov chains: Hitting and mixing times}
\label{subsec:quantities-related-to-mc}
In this subsection, we define certain quantities related to Markov chains which we shall use in subsequent sections for our analysis.
\\~\\
\textbf{Spatial search problem and hitting time:} Consider a graph $G(X,E)$ with $|X|=n$ vertices and $|E|=e$ edges. Consider a subset $M \subset X$ of vertices that are marked. Then the spatial search problem involves finding any of the marked vertices in $M$. This problem can be solved by both classical random walks and quantum walks.

Given an ergodic, reversible Markov chain $P$ with a stationary state $\pi$, the random walk based algorithm to solve the spatial search problem is described in Algorithm \ref{algo:spatial-search-by-random-walk}.
\RestyleAlgo{boxruled}
\begin{algorithm}[ht]
\caption{Spatial search by random walk}\label{algo:spatial-search-by-random-walk}
\begin{itemize}
\item[1.~] Sample a vertex $x\in X$ from the stationary state $\pi$ of $P$.
\item[2.~] Check if $x\in M$.
\item[3.~] If $x$ is marked, output $x$.
\item[4.~] Otherwise update $x$ according to $P$ and go to step 2.
\end{itemize}
\end{algorithm}
The \textit{hitting time} of $P$ with respect to $M$ is the expected number of times step $4$ of Algorithm \ref{algo:spatial-search-by-random-walk} is executed. Let us denote this by $HT(P,M)$. Thus, the random walk based algorithm finds a marked vertex in time $\mathcal{O}\left(HT(P,M)\right)$. Note that the random walk algorithm stops as soon as a marked element is reached. Thus, this is equivalent to applying an absorbing Markov chain $P'$ that is obtained by replacing all the outgoing edges from the marked vertices of $P$ by self loops. From this we can define $HT(P,M)$.
~\\~\\
\textbf{Hitting time of a Markov chain:} The hitting time of any Markov chain $P$ with respect to a set of marked elements $M$ can be expressed as 
\begin{equation}
\label{eqmain:hitting_time}
HT(P,M)=\sum_{j=1}^{n-m}\dfrac{|\braket{v'_j}{U}|^2}{1-\lambda'_j},
\end{equation}
where $\lambda'_j$ and $\ket{v'_j}$ are the eigenvalues and eigenvectors of the matrix $D(P')$ and
$$\ket{U}=\dfrac{1}{\sqrt{1-p_M}}\sum_{x\neq M}\sqrt{\pi_x}\ket{x},$$
where $p_M$ is the probability of sampling a marked vertex from the stationary state of $P$.
~\\~\\
\textbf{Interpolated hitting time and Extended hitting time:} For any interpolated Markov chain $P(s)$, in refs.~\cite{krovi2016quantum,chakraborty2018finding}, the authors define a quantity known as the \textit{interpolated hitting time} in the context of spatial search which will also be useful here for subsequent analysis. This is defined as
\begin{equation}
\label{eqmain:interpolated-hitting-time}
HT(s)=\sum_{j=1}^{n-1}\dfrac{|\braket{v_j(s)}{U}|^2}{1-\lambda_j(s)}.
\end{equation}
There is a relationship between the spectral gap of the Markov chain and $HT(s)$ since 
\begin{equation}
\label{eqmain:iht-spectral-gap}
HT(s)\leq \dfrac{1}{\Delta(s)}\sum_{j=1}^{n-1}|\braket{v_j(s)}{U}|^2.
\end{equation}
For the spatial search algorithm, we shall find that the quantity of interest is the \textit{extended hitting time}. The extended hitting time of $P$ with respect to a set $M$ of marked elements is given by
\begin{equation}
\label{eqmain:extended-hitting-time}
HT^+(P,M)=\lim_{s\rightarrow 1} HT(s),
\end{equation}
Clearly for $|M|=1$, we have that $HT^+(P,M)=HT(P,M)$. Krovi et al. proved an explicit relationship between $HT(s)$ and $HT^+(P,M)$ \cite{krovi2016quantum}. They showed that
\begin{equation}
\label{eqmain:interpolated-vs-extended-hitting-time}
HT(s)=\dfrac{p_M^2}{\left(1-s(1-p_M)\right)^2}HT^+(P,M).
\end{equation} 
Combining Eqs.~\eqref{eqmain:iht-spectral-gap} and \eqref{eqmain:interpolated-vs-extended-hitting-time}, we have
\begin{equation}
\label{eqmain:hitting-time-vs-gap}
HT^+(P,M)\leq \dfrac{1}{\Delta(s)}.\dfrac{\left(1-s(1-p_M)\right)^2}{p_M^2}\sum_{j=1}^{n-1}|\braket{v_j(s)}{U}|^2.
\end{equation} 
~\\
\textbf{Mixing-time of a Markov chain:} Given a reversible Markov chain $P$, any initial probability distribution over the state space converges to the stationary distribution $\pi$, i.e.\ $\lim_{t\rightarrow\infty}\mu=\pi$, for any initial distribution $\mu$. Given $P$ and an initial state $\mu$, the mixing-time of a classical random walk is defined as the minimum time $T_{\mathrm{mix}}$ such that $\forall t\geq T_{\mathrm{mix}}$ we have that
$$\dfrac{1}{2}\|\mu P^t-\pi\|_1\leq \epsilon,$$
for some $\epsilon\in (0,1)$, where $\frac{1}{2}\|.\|_1$ is the total variation distance. 

That is, $T_{\mathrm{mix}}$ is the minimum time required for the Markov chain to converge to a distribution that is $\epsilon$-close to the stationary distribution which implies that \cite{aldous1982some}:
~\\
\begin{equation}
\label{eqmain:mixing-time-definition}
T_{\mathrm{mix}}\leq \dfrac{1}{\Delta}\log\left(\dfrac{1}{\epsilon\pi_{\min}}\right),
\end{equation}
where $\Delta$ is the spectral gap of $P$ and $\pi_{\min}=\min_x \pi_x$.
~\\
Thus, given an ergodic, reversible Markov chain $P$ with stationary state $\pi$ and spectral gap $\Delta$, one can sample from a distribution that is $\epsilon$-close to $\pi$ in time $\bigOt{1/\Delta}$. Next we discuss how one can use von Neumann measurements to prepare eigenstates of Hamiltonians, a tool which will help us provide an analog quantum algorithm for solving $QSSamp$.
\section{Quantum state generation by von Neumann measurements} 
\label{sec-main:von-neumann}
In this section, we make use of von Neumann measurements to prepare eigenstates of a Hamiltonian. The goal would be to use this technique to prepare the eigenstate of the quantum walk Hamiltonian (encoding an ergodic reversible Markov chain $P$) that corresponds to a coherent encoding of the stationary distribution of $P$.

In this framework, in order to measure any observable $\widehat{O}$, the system of interest is coupled to a pointer, which is simply a free particle in one dimension. If $H$ represents the Hamiltonian of the system and $\widehat{p}$ the momentum operator corresponding to the pointer, then the total Hamiltonian corresponding to the coupling between the system and the pointer is given by
\begin{equation}
\widetilde{H}=H+\dfrac{\widehat{p}^2}{2m}+g\ \widehat{O}\otimes \widehat{p},
\end{equation}
where $m$ is the mass of the free particle and $g$ is the interaction strength between the observable and the pointer. Since we are interested in measuring the energy of the system, we have $\widehat{O}=H$. We consider the particle as ``massive", thereby enabling us to neglect the free Hamiltonian of the particle. Furthermore, we assume that we are working with units such that the interaction strength $g=1$. These imply that 
\begin{equation}
\widetilde{H}= H\otimes \widehat{p}.
\end{equation}
It is well known that the momentum operator, $\widehat{p}=-i\frac{d}{dx}$ is a generator of translation in the position of the particle. In other words, the operator $e^{-ix_0\widehat{p}}$ applied to a wavepacket whose wavefunction is $\psi(x)$ results in 
\begin{align}
e^{-ix_0\widehat{p}}\psi(x)&=e^{-x_0\frac{d}{dx}}\psi(x)\\
							  &=\left(I-x_0\frac{d}{dx}+...\right)\psi(x)\\
							  &=\psi(x-x_0).
\end{align}
Thus the wavepacket is translated in position by $x_0$. Now consider that the system Hamiltonian $H$ has eigenvalues 
$$
\lambda_n=0<\Delta=\lambda_{n-1}\leq \lambda_{n-2}\leq\cdots \lambda_1\leq 1,
$$ 
such that $H\ket{v_j}=\lambda_j\ket{v_j}$. Furthermore, suppose that we initialize the pointer to a state $\ket{x=0}$, a wavepacket centred around $0$. Then, 
\begin{equation}
e^{-i\widetilde{H}t}\ket{v_j}\ket{x=0}=\ket{v_j}\ket{x=\lambda_jt}.
\end{equation} 

That is, the wavepacket is translated in position by $\lambda_j t$ and as such, measuring the displacement of the pointer register can in principle reveal information about the eigenstate of $H$ in the first register. By linearity, for any initial state $\ket{\psi_0}=\sum_{j=1}^n\alpha_j\ket{v_j}$, we have
\begin{align}
e^{-i\widetilde{H}t}\ket{\psi_0}\ket{x=0}&=e^{-iH\widehat{p}t}\sum_{j=1}^n\ket{\psi_0}\ket{x=0}\\
										 &=\sum_{j=1}^n\alpha_j\ket{v_j}\ket{x=\lambda_j t}.
\end{align}
In order to implement this on a quantum computer, we assume that the pointer register is of $l$ qubits. The choice of $l$ is crucial as it determines the precision up to which the position of the pointer is obtained. In fact, if we measure the position of the pointer with a high enough precision to resolve all eigenvalue gaps, $(\lambda_i-\lambda_j)t$, a measurement of the position of the pointer results in a measurement of the system Hamiltonian $H$.

For our purposes, we shall show how this formalism can be used to prepare the $0$-eigenstate of $H$, i.e.\ $\ket{v_n}$, in a purely analog fashion. To that end, we formally state via Lemma \ref{lem_main:prep-state-von-neumann} and Corollary \ref{cor_main:prep-state-high-accuracy}.
~\\
\begin{lemma}
\label{lem_main:prep-state-von-neumann}
Let $H$ be a Hamiltonian with eigenvalues $\lambda_n=0<\Delta=\abs{\lambda_{n-1}}\leq\cdots\abs{\lambda_1}\leq 1$ such that $H\ket{v_j}=\lambda_j\ket{v_j}$. Let $\widehat{p}$ represent the momentum operator corresponding to a free particle in one dimension with its mass large enough so that its free Hamiltonian can be neglected and so that it can be represented in $l$ qubits as 
$$\widehat{p}=\sum_{q=0}^{2^l-1}\dfrac{q}{2^l}\ket{q}\bra{q},$$
where 
\begin{equation}
l=\lceil \log_2(\tau/\pi) \rceil.
\end{equation}
for some $\tau>0$.
Furthermore let 
$$\ket{\psi_0}=\sum_{j=1}^n{\alpha_j}\ket{v_j}.$$ 
Then starting from the state $\ket{\psi_0}\ket{x=0}$ and evolving for a time $\tau$
according to the Hamiltonian $\widetilde{H}=H\otimes \widehat{p}$, results in a state 
$$\ket{\widetilde{\psi}}=\alpha_n\ket{v_n}\ket{0}+\sum_{k=1}^{n-1}\alpha_k\ket{v_k}\left(\gamma_k\ket{0}+\Gamma_k\ket{\Gamma_k}\right),$$
where $|\gamma_k|\leq\pi/(\abs{\lambda_k}\tau),~|\Gamma_k|=\sqrt{1-|\gamma_k|^2}$ and $\braket{\Gamma_k}{0}=0$ for $1\leq k\leq n-1$.

In particular, for
$$\tau=\dfrac{2\pi}{\Delta},$$ 
we have $|\gamma_k|\leq1/2$ and $|\Gamma_k|\geq\sqrt{3}/{2}$ for $1\leq k\leq n-1$.
\end{lemma}
~\\
\begin{proof}

 If $\ket{q}$ represents the momentum eigenstates, then the momentum operator is represented by
\begin{equation}
\widehat{p}=\sum_{q=0}^{2^l-1}\dfrac{q}{2^l}\ket{q}\bra{q}.
\end{equation}
Note that the position and momentum states are equivalent up to a Fourier transform and so the localized wavepacket centred at $x=0$ is completely delocalized in the momentum basis. That is,
\begin{equation}
\label{eqmain:delocalized-wavepacket}
\ket{x=0}=\dfrac{1}{\sqrt{2^l}}\sum_{q=0}^{2^l-1}\ket{q}.
\end{equation}

Therefore,
\begin{align}
\label{eq-main:pointer-evolution}
e^{-i(H\otimes \widehat{p})\tau}\ket{\psi_0}\ket{x=0}&=e^{-i(H\otimes \widehat{p})\tau}\ket{\psi_0}\left(\dfrac{1}{\sqrt{2^l}}\sum_{q=0}^{2^l-1}\ket{q}\right)\\
	                    &=\sum_{k=1}^n \alpha_k\ket{v_k}\left(\dfrac{1}{\sqrt{2^l}}\sum_{q=0}^{2^l-1}e^{\frac{-i\lambda_k\tau q}{2^l}}\ket{q}\right).
\end{align}  
Since we ultimately want to read off the position of the pointer variable, we re-express the pointer register in the position-basis to obtain
\begin{align}
\begin{split}
&e^{-i(H\otimes \widehat{p})\tau}\ket{\psi_0}\ket{x=0}=\\
&\sum_{k=1}^n \alpha_k\ket{v_k}\left(\dfrac{1}{2^l}\sum_{x=0}^{2^l-1}\sum_{q=0}^{2^l-1}e^{\frac{i(x-\lambda_k\tau) q}{2^l}}\ket{x}\right).
\end{split}
\end{align}
The pointer register has a measure of the displacement of the wavepacket which was initially centered at $x=0$. In fact, as shown previously, the shift will be proportional to the eigenvalue corresponding to the eigenstate in the first register (expressed in $l$ qubits). That is, we will have states of the form $\ket{v_j}\ket{\lambda_j\tau}$. We are interested in preparing the $0$-eigenstate $\ket{v_n}$. We first observe that the amplitude of obtaining $\ket{0}$ in the pointer register when the first register is in the state $\ket{v_n}\ket{0}$ is one, i.e.\ 
$$ e^{-i\tau(H\otimes\widehat{p})}\ket{v_n}\ket{0}\mapsto\ket{v_n}\ket{0}.$$
On the other hand, for any other eigenstate $\ket{v_k}$, the amplitude corresponding to measuring $\ket{0}$ in the second register is
\begin{align}
\dfrac{1}{2^l}\left|\sum_{q=0}^{2^l-1}e^{\frac{i(x-\lambda_k\tau) q}{2^l}}\right|&=\dfrac{1}{2^l}\left|\sum_{q=0}^{2^l-1}e^{\frac{-i\lambda_k\tau q}{2^l}}\right|\\
								&=\dfrac{1}{2^l}\left|\dfrac{1-e^{-i\lambda_k\tau}}{1-e^{-i\lambda_k\tau/2^l}}\right|,
\end{align}
Since $|1-e^{-iz}|\leq 2$ for any $z$ and $|1-e^{-iz}|\geq 2|z|/\pi$ for any $z\in[-\pi,\pi]$, we can bound this amplitude as
\begin{align}
\dfrac{1}{2^l}\left|\sum_{q=0}^{2^l-1}e^{\frac{i(x-\lambda_k\tau) q}{2^l}}\right|&\leq\frac{\pi}{\abs{\lambda_k}\tau}
\end{align}
where we have used the fact that $l\geq \log_2(\tau/\pi)$ so $\abs{\lambda_k}\tau/2^l\leq\pi$.

Thus, after the time evolution for a time $\tau$, the state of the system and the pointer is given by
\begin{equation}
\label{eq:final-state-ham-evolution}
\ket{\widetilde{\psi}}=\alpha_n\ket{v_n}\ket{0}+\sum_{k=1}^{n-1}\alpha_k\gamma_k\ket{v_k}\ket{0}+\sum_{k=1}^{n-1}\alpha_k\Gamma_k\ket{v_k}\ket{\Gamma_k},
\end{equation}
where $|\gamma_k|\leq\pi/(\abs{\lambda_k}\tau),~|\Gamma_k|=\sqrt{1-|\gamma_k|^2}$ and $\braket{\Gamma_k}{0}=0$ for $1\leq k\leq n-1$.

Since $\abs{\lambda_k}\geq\Delta$, for $\tau=2\pi/\Delta$ we have  $|\gamma_k|\leq 1/2$ and $|\Gamma_k|\geq\sqrt{3}/{2}$ for $1\leq k\leq n-1$.
\end{proof}
~\\
We shall use this lemma to derive the following corollary.
~\\
\begin{corollary}
\label{cor_main:prep-state-high-accuracy}
Let $\epsilon'\leq\epsilon |\alpha_n|/\sqrt{2}$ where $\epsilon\in (0,1)$ and suppose that the pointer register contains 
\begin{equation}
\label{eqmain:total-qubits-pointer}
m=l\ .\ \lceil \log_2(1/\epsilon')\rceil
\end{equation}
qubits initialized in the state $\ket{x=0}^{\otimes m}$, where $l=\lceil\log_2(\tau/\pi)\rceil$.

Then repeating the Hamiltonian evolution of Lemma \ref{lem_main:prep-state-von-neumann} with $\tau=\dfrac{2\pi}{\Delta}$ a total of $\lceil \log(1/\epsilon')\rceil$-times using a fresh block of $l$ pointer qubits each time, followed by post-selecting on the pointer register to be in $\ket{0}^{\otimes{m}}$, succeeds with probability at least $|\alpha_n|^2$ in constructing a quantum state $\ket{\phi}$ such that
$$
\nrm{\ket{v_n}-\ket{\phi}}_2\leq \epsilon,$$
in time
\begin{equation}
\label{eqmain:time-evolution-pointer}
T=\Theta\left(\dfrac{1}{\Delta}\log\left(\frac{1}{\epsilon|\alpha_n|}\right)\right).
\end{equation}
\end{corollary}
~\\
\begin{proof}
After the application of $e^{-i(H\otimes \widehat{p})\tau}$ a total of $\lceil\log(1/\epsilon')\rceil$-times using $l$ blocks of qubits in the pointer register each time, observe that for any $k\neq n$, the amplitude for observing $\ket{0}^{\otimes m}$ in the pointer register when there is $\ket{v_k}$ in the first register is bounded by
\begin{equation}
|\epsilon_k|=|\gamma_k|^{\lceil\log_2(1/\epsilon')\rceil} \leq \left(\dfrac{1}{2}\right)^{\lceil\log_2(1/\epsilon')\rceil}\leq \epsilon'.
\end{equation}
This implies that the resulting state after this procedure is given by
\begin{align}
\label{eq-main:psi-f}
\begin{split}
&\ket{\psi_f}=\alpha_n\ket{v_n}\ket{0}^{\otimes m}+\sum_{k=1}^{n-1}\epsilon_k\alpha_k\ket{v_k}\ket{0}^{\otimes m}\\
&~~~+\sum_{k=1}^{n-1}\alpha_k\delta_k\ket{v_k}\ket{\Gamma_k^{(m)}},
\end{split}
\end{align}
where $0\leq \epsilon_k\leq \epsilon'$ and $\sqrt{1-\epsilon'^2}\leq \delta_k\leq 1$. This takes time
$$T=\frac{2\pi}{\Delta}\lceil\log(1/\epsilon')\rceil
=\Theta\left(\dfrac{1}{\Delta}\log\left(\frac{1}{|\alpha_n|\epsilon}\right)\right).$$
The state in Eq.~\eqref{eq-main:psi-f} can be re-written as
\begin{equation}
\ket{\psi_f}=\alpha_n\left(\ket{v_n}+\ket{\text{err}}\right)\ket{0}^{\otimes m}+\sum_{k=1}^{n-1}\alpha_k\delta_k\ket{v_k}\ket{\Gamma_k^{(m)}},
\end{equation}
where the state
$$\ket{\text{err}}=\sum_{k=1}^{n-1}\dfrac{\epsilon_k\alpha_k}{\alpha_n}\ket{v_k}$$
has norm $\delta$ with
$$
\delta^2:=\|\ket{\text{err}}\|^2=\sum_{k=1}^{n-1}\left|\dfrac{\epsilon_k\alpha_k}{\alpha_n}\right|^2\leq \dfrac{\epsilon'^2}{|\alpha_n|^2}=\frac{\epsilon^2}{2}.
$$
This implies that post-selecting on obtain $\ket{0}^{\otimes m}$ in the pointer register we obtain the state 
\begin{equation}
\ket{\phi}=\frac{1}{\sqrt{1+\delta^2}}\left[\ket{v_n}+\ket{\text{err}}\right],
\end{equation}
with probability $|\alpha_n|^2(1+\delta^2)\geq|\alpha_n|^2$, such that
\begin{equation}
\nrm{\ket{v_n}-\ket{\phi}}^2=\left(1-\frac{1}{\sqrt{1+\delta^2}}\right)^2+\delta^2\leq 2\delta^2\leq\epsilon^2
\end{equation} 
\end{proof}
Thus Lemma \ref{lem_main:prep-state-von-neumann} and corollary \ref{cor_main:prep-state-high-accuracy} can be used to prepare the eigenstate $\ket{v_n}$ with success probability at least $|\alpha_n|^2$, and repeating this procedure $\Theta(1/|\alpha_n|^2)$ times allows to boost this probability to $\Theta(1)$, which leads to a protocol using total time
$$T=\Theta\left(\dfrac{1}{\Delta|\alpha_n|^2}\log\left(\frac{1}{\epsilon|\alpha_n|}\right)\right).$$
Note that it would have been possible to use quantum amplitude amplification to reduce quadratically the dependency on $|\alpha_n|$. However, we are interested in developing analog algorithms, assuming that we have access to a time-independent Hamiltonian. Provided that the cost of preparing the initial state $\ket{\psi_0}$ is small, the cost of the algorithm is the total time of Hamiltonian evolution. Moreover our protocol (Sec.~\ref{sec:preparation-of-stationary-state-von-neumann}) to prepare the stationary state of any reversible Markov chain ensures that $|\alpha_n|=\Theta(1)$, thereby resulting in at most a constant slowdown with respect to amplitude amplification.


\section{Hamiltonian for quantum walk on any reversible Markov chain} 
\label{sec-main:hamiltonian-any-markov-chain}
Given any ergodic, reversible Markov chain, we shall make use of the Hamiltonian introduced by Somma and Ortiz \cite{somma2010quantum} and subsequently used in Refs.~\cite{krovi2010adiabatic, chakraborty2018finding}. We recall the Hamiltonian and its spectral properties here for completeness and it will be used in our quantum algorithm for $QSSamp$.
\subsection{Defining the Hamiltonian}
\label{subsec:hamiltonian-definition}
Let $p_{xy}(s)$ denote the $(x,y)^{\mathrm{th}}$-entry of $P(s)$ and let $E$ be the set of edges of $P(s)$. Furthermore let $\mathcal{H}=\mathrm{span}\{\ket{x}: x\in X \}$ . Then one can define a unitary $V(s)\in \mathcal{H} \times \mathcal{H} $ such that for all $x\in X$,
\begin{equation}
\label{eqmain:unitary-for-hamiltonian}
V(s)\ket{x,0}=\sum_{y\in X}\sqrt{p_{xy}(s)}\ket{x,y},
\end{equation}
where the state $\ket{0}$ represents a fixed reference state in $\mathcal{H}$. Let us also define the swap operator 
$$ S\ket{x,y}=
\begin{cases}
\ket{y,x}, & \text{if $(x,y)\in E$} \\
\ket{x,y}, & \text{otherwise}.
\end{cases}
$$
Observe that $\braket{x,0|V(s)^\dag S V(s)}{y,0}=\sqrt{p_{yx}(s)p_{xy}(s)}=D_{xy}(P(s))$. Then, if $\Pi_0=I\otimes \ket{0}\bra{0}$, we have,
\begin{equation}
V^\dag (s)SV(s)\Pi_0\ket{y,0}=\sum_{x\in X}\sqrt{p_{yx}(s)p_{xy}(s)}\ket{x,0}+\ket{\Phi}^\perp,
\end{equation}
so that $\Pi_0\ket{\Phi}^\perp=0$. We define the search Hamiltonian as
\begin{equation}
\label{eqmain:search_hamiltonian}
H(s)=i[V(s)^\dag S V(s),\Pi_0].
\end{equation}

In Ref.~\cite{chakraborty2018finding}, we have shown that $H(s)$, in a rotated basis, corresponds to a quantum walk on the edges of $P(s)$. That is, the rotated Hamiltonian
\begin{align}
\overline{H}(s)&=V(s)H(s)V(s)^\dag\\
            &=i[S,V\Pi_0V^\dag],
\end{align}
corresponds to a quantum walk on the edges of $P(s)$. If the walker is localized in a directed edge from node $x$ to node $y$, i.e.\ $\ket{x,y}$, then the walker can move to a superposition of outgoing edges from node $y$ of the form $\ket{y,.}$. Note that our algorithms (See Algorithm \ref{algo:search-ctqw} and Algorithm \ref{algo1}) could be implemented using the Hamiltonian $\overline{H}(s)$ instead of $H(s)$. In such a case, we need to apply the same rotation to the initial state of the algorithm and the final state of the algorithm. However, subsequently we shall be working with $H(s)$ as it simplifies the analysis considerably. In the next subsection, we will characterize the spectrum of $H(s)$.
\subsection{Spectrum of $H(s)$}
\label{subsec:spectrum-somma-ortiz}
As discussed in Sec.~\ref{subsec:basics-mc}, the spectrum of $H(s)$ is related to that of $D(P(s))$ and in particular, the state $\ket{v_n(s),0}$ is an eigenstate of $H(s)$ with eigenvalue zero. The spectrum of $H(s)$ has been explicitly described in Ref.~\cite{krovi2010adiabatic} and we mention it here for completeness. The total Hilbert space of $H(s)$ can be divided into the following set of invariant subspaces: 

For $1\leq k\leq n-1$,
\begin{align}
\mathcal{B}_k(s)&=\mathrm{span}\{\ket{v_k(s),0},V(s)^\dag S V(s)\ket{v_k(s),0}\},\\
\mathcal{B}_n(s)&=\mathrm{span}\{\ket{v_n(s),0}\}\\
\mathcal{B}^\perp(s)&=(\oplus_{k=1}^n\mathcal{B}_k)^\perp.
\end{align}

Now, observe that
\begin{align}
\Pi_0 V(s)^\dag SV(s)\ket{v_n(s),0}=\ket{v_n(s),0}\\
V(s)^\dag SV(s)\Pi_0\ket{v_n(s),0}=\ket{v_n(s),0}.
\end{align}
This implies 
\begin{equation}
H(s)\ket{v_n(s),0}=0,
\end{equation}
i.e.\ $\ket{v_n(s),0}$ is an eigenstate with eigenvalue $0$.
\\~\\
On the other hand, note that for $1\leq k \leq n-1$, the eigenstates and eigenvalues of $H(s)$ in $\mathcal{B}_k(s)$ are
\begin{align}
\label{eqmain:eigenpair-somma-ortiz}
\ket{\Psi^\pm_k(s)}&=\dfrac{\ket{v_k(s),0}\pm i\ket{v_k(s),0}^\perp}{\sqrt{2}},\\
E^\pm_k(s)&=\pm \sqrt{1-\lambda_k(s)^2},\label{eqmain:eigenvalues-somma-ortiz}
\end{align}
where $\ket{v_k(s),0}^\perp$ is a quantum state that is in $\mathcal{B}_k(s)$ such that $\Pi_0\ket{v_k(s),0}^\perp=0$. Thus, if the underlying Markov chain has a spectral gap $\Delta(s)$, then in this subspace $H(s)$ has a quadratically amplified spectral gap given by  
\begin{equation}
|E_n(s)-E^{\pm}_{n-1}(s)|=\sqrt{1-\lambda^2_{n-1}(s)}=\Theta\left(\sqrt{\Delta(s)}\right).
\end{equation}
Now, there are $n^2$ eigenvalues of $H(s)$ out of which $2n-1$ belong to $\mathcal{B}_k(s)\cup\mathcal{B}_n(s)$. The remaining $(n-1)^2$ eigenvalues are $0$ and belong to $\mathcal{B}^\perp(s)$ which is the orthogonal complement of the union of the invariant subspaces. We need not worry about this subspace as we start from a state that has no support on $B^\perp (s)$ which is an invariant subspace of $H(s)$. Thus, throughout the evolution under $H(s)$, our dynamics will be restricted to $\mathcal{B}_k(s)\cup\mathcal{B}_n(s)$.

The following connection, first observed in~\cite{krovi2010adiabatic}, between the eigenstates of $H(s)$ in $\mathcal{B}_k(s)\cup\mathcal{B}_n(s)$ and the interpolated hitting time $HT(s)$ introduced in Sec.~\ref{subsec:quantities-related-to-mc} will be useful to the complexity analysis of our algorithms
\begin{prop}\label{prop-main:ht-connection}
 Let $h(s)$ be defined as
 \begin{equation}
     h(s)=\sum_{\sigma=\pm 1}\sum_{k=1}^{n-1}\frac{\abs{\braket{\Psi_k^\sigma(s)}{U,0}}^2}{\abs{E_k^\sigma(s)}^2}
 \end{equation}
 Then, we have
 \begin{equation}
     \frac{1}{2}HT(s)\leq h(s)\leq HT(s).
 \end{equation}
\end{prop}
\begin{proof}
From Eqs.~(\ref{eqmain:eigenpair-somma-ortiz}-\ref{eqmain:eigenvalues-somma-ortiz}), we have
\begin{align*}
    \abs{\braket{\Psi_k^\sigma(s)}{U,0}}^2&=\frac{\abs{\braket{v_k^\sigma(s)}{U}}^2}{2}\\
    \abs{E_k^\sigma(s)}^2&=1-\lambda_k(s)^2
\end{align*}
so that
\begin{equation*}
h(s)=\sum_{k=1}^{n-1}\dfrac{|\braket{v_k(s)}{U}|^2}{1-\lambda_k(s)^2}.
\end{equation*}
The bounds on $h(s)$ then follow from the definition of $HT(s)$ and the fact that $$1-\lambda_k(s)\leq1-\lambda_k(s)^2=(1-\lambda_k(s))(1+\lambda_k(s))\leq 2(1-\lambda_k(s))$$
\end{proof}

\section{Spatial search by continuous-time quantum walk using von-Neumann measurements}
\label{sec-main:search-pointer}
We first show how to make use of the state-generation scheme described in Sec.~\ref{sec-main:von-neumann} to provide a continuous-time quantum walk based algorithm to solve the spatial search problem. This algorithm will provide an intuitive understanding of our analog quantum algorithm for $QSSamp$.
 
Suppose we are given an ergodic, reversible Markov chain $P$ with the marked set denoted by $M\subset X$. The spatial search algorithm on $P$ involves finding a node within this marked set and is often tackled by the formalism of random walks. We have seen previously in Sec.~\ref{subsec:quantities-related-to-mc} that the expected number of steps taken by the walker to find a node within this marked set is known as the \textit{hitting time} of $P$ with respect to $M$. Quantum walks provide a natural framework to tackle this problem. A natural question to ask is whether a quantum walk can offer any speedup over its classical counterpart in order to solve the spatial search problem. Here, we concentrate on the continuous-time quantum walk framework to tackle this problem.

The spatial search algorithm by continuous-time quantum walk on $P$ involves evolving a time-independent Hamiltonian (which encodes the connectivity of $P$), starting from some initial state, for some time, and then measuring in the basis spanned by the states of $P$. 

Childs and Goldstone \cite{childs2004spatial} introduced the first continuous-time quantum walk-based algorithm to tackle the spatial search problem for simple, unweighted graphs. They showed that this algorithm, defined as a quantum walk on the nodes of the underlying graph, could find a marked node in $\mathcal{O}(\sqrt{n})$ time for certain graphs with $n$ nodes such as the complete graph, hybercube and $d$-dimensional lattices with $d>4$. This offered a quadratic speedup over classical random walks for the spatial search problem on these graphs. When $d=4$, the running time of the Childs and Goldstone algorithm is $\mathcal{O}(\sqrt{n}\log n)$, offering a less than quadratic speedup whereas there is no substantial speedup for $d<4$. Since then, a plethora of results have been published exhibiting a $\mathcal{O}(\sqrt{n})$ running time on certain specific graphs each requiring an ad-hoc analysis~\cite{janmark2014global, foulger2014quantum, childs2014spatial-crystal, meyer2015connectivity, novo2015systematic, philipp2016continuous, wong2016quantum, chakraborty2016spatial, chakraborty2017optimal}. Recently in Ref.~\cite{chakraborty2020optimality}, the authors provided the necessary and sufficient conditions for this algorithm to be optimal under very general conditions on the spectrum of the quantum walk Hamiltonian. They showed that attaining a generic quadratic speedup is impossible using this algorithm. 

In Ref.~\cite{chakraborty2018finding}, the authors provided a different spatial search algorithm by continuous-time quantum walk which finds a marked element on any ergodic, reversible Markov chain in square-root of the \textit{extended hitting time}, thereby matching the running time of best known algorithms in the DTQW-framework in the case of where a single node is marked, i.e.\ $|M|=1$ \cite{krovi2016quantum, ambainis2019quadratic}. Given any $P$, their algorithm made use of the framework of interpolating quantum walks
$$
P(s)=(1-s)P+sP',
$$
where $P'$ is the absorbing Markov chain such that any state having support only over the marked vertices is its stationary state. In fact, they used the Somma-Ortiz Hamiltonian $H(s)$ (described in Sec.~\ref{sec-main:hamiltonian-any-markov-chain}), to define quantum walk on the edges of $P(s)$. The underlying technique is to use a procedure called quantum phase randomization to (approximately) prepare a (mixed) state that has a constant overlap with the $0$-eigenstate of $H(s)$. For some specific value of $s$, this eigenstate has a constant overlap with the marked subspace $M$. This  required measurement at a time chosen uniformly at random between $[0,~\sqrt{HT^+(P,M)}]$, where $HT^+(P,M)$ is the \textit{extended hitting time}, defined in Eq.~\eqref{eqmain:extended-hitting-time}.

In this section, we provide an alternative spatial search algorithm (Algorithm \ref{algo:search-ctqw}) by continuous-time quantum walk that finds an element from a marked set $M$ in time that scales as the square root of the \textit{extended hitting time}. Algorithm \ref{algo:search-ctqw} is similar in spirit to that of Ref.~\cite{chakraborty2018finding} in that both make use of the Somma-Ortiz Hamiltonian $H(s)$ defined in Eq.~\eqref{eqmain:search_hamiltonian}. 

\RestyleAlgo{boxruled}
\begin{algorithm}[ht]
  \caption{Spatial search by continuous-time quantum walk}\label{algo:search-ctqw}
  Let $\delta\in (0,1/4)$ and  $\tau=\frac{\pi}{\delta}\sqrt{2HT(s^*)}$.
\begin{itemize}
\item[1.~] Prepare the state $\ket{\psi_0}=\ket{v_n(0),0}\ket{x=0}$ and check if the first register is a marked node. If so, go to step 3, otherwise continue.
\item[2.~] Evolve according to $H(s^*)\otimes\widehat{p}$ for time $\tau$ starting from the state $\ket{\psi_0}$.
\item[3.~] Measure in the basis of the state-space of the Markov chain in the first register.
\end{itemize} 
 \end{algorithm}

However, motivated by the problem of quantum state generation using von Neumann measurements, Algorithm \ref{algo:search-ctqw} prepares the $0$-eigenstate of $H(s)$ by coupling this Hamiltonian to a free-particle in one dimension. Unlike the algorithm of Ref.~\cite{chakraborty2018finding}, we evolve the Hamiltonian for a fixed time $\tau=\frac{\pi}{\delta}\sqrt{2HT(s^*)}=\Theta(\sqrt{HT^+(P,M)})$, where
\begin{equation}
\label{eq:s*}
s^*=1-p_M/(1-p_M).
\end{equation}

Furthermore, the $0$-eigenstate $\ket{v_n(s^*)}$ can be written as
\begin{equation}
\label{eq-main:highest-eigenstate-somma-ortiz-at-s*}
\ket{v_n(s^*)}=\dfrac{\ket{U}+\ket{M}}{\sqrt{2}},
\end{equation}
where $\ket{U}$ and $\ket{M}$ are as defined in Proposition \ref{prop-main:zero-estate-somma-ortiz}. Thus it has a constant overlap with both $\ket{U}$ and $\ket{M}$. The intuition behind the algorithm is as follows. In step 1, it either detects a marked node, in which case it succeeds with probability 1, or it prepares state $\ket{U,0}$. Since this state has overlap $1/\sqrt{2}$ with $\ket{v_n(s^*),0}$, step 2 allows to prepare a state such that, if we were to project the third register on $\ket{0}$, we would obtain with probability close to $1/2$ a state close to $\ket{v_n(s^*),0}$. This state therefore has overlap close to $1/\sqrt{2}$ over $\ket{M,0}$, so that step 3 will yield a marked node with probability close to $1/2$.

We now formally state Algorithm \ref{algo:search-ctqw} and prove its correctness in Lemma \ref{lem_main:correctness-search-ctqw}.
~\\
\begin{lemma}
\label{lem_main:correctness-search-ctqw}
Algorithm \ref{algo:search-ctqw} outputs a marked node with probability at least $1/4-\delta$ in time
$$T=\Theta\left(\sqrt{HT^+(P,M)}\right).$$
\end{lemma}
~\\
\begin{proof}
We shall make use of Lemma \ref{lem_main:prep-state-von-neumann}. Observe that for $s=s^*=1-p_M/(1-p_M)$, the $0$-eigenstate of $H(s^*)$ is simply $\ket{v_n(s^*),0}$ with
\begin{equation}
\ket{v_n(s^*)}=\dfrac{\ket{U}+\ket{M}}{\sqrt{2}},
\end{equation}
where $\ket{U}$ and $\ket{M}$ are as defined in Proposition \ref{prop-main:zero-estate-somma-ortiz}. Also, the initial state in the first register is
\begin{equation}
\ket{v_n(0)}=\sqrt{1-p_M}\ket{U}+\sqrt{p_M}\ket{M},
\end{equation}
so at the end of step 1, we project on $\ket{M}$ with probability $p_M$ (in which case we will immediately obtain a marked node in step 3), or we project on $\ket{U}$.

From Lemma \ref{lem_main:prep-state-von-neumann},
the time evolution of the Hamiltonian $\widetilde{H}(s^*)=H(s^*)\otimes \widehat{p}$, starting from the state
\begin{align*}
\ket{\psi_0}&=\ket{U,0}\ket{x=0}\\
&=\alpha_n\ket{v_n(s^*),0}\ket{0}+\sum_{\sigma=\pm 1}\sum_{k=1}^{n-1}\alpha_k^\sigma\ket{\Psi_k^{\sigma}(s^*)},
\end{align*}
where $\alpha_n=\braket{v_n(s^*)}{U}=\frac{1}{\sqrt{2}}$ and $\alpha_k^\sigma=\braket{\Psi_k^\sigma(s^*)}{U,0}$,
prepares a state
\begin{align*}
    \ket{\psi_1}=&\alpha_n\ket{v_n(s^*),0}\ket{0}\\
    &+\sum_{\sigma=\pm 1}\sum_{k=1}^{n-1}\alpha_k^\sigma\ket{\Psi_k^{\sigma}(s^*)}\left(\gamma_k^\sigma\ket{0}+\Gamma_k^\sigma\ket{\Gamma_k^\sigma}\right),
\end{align*}
where $|\gamma_k^\sigma|\leq\pi/(\abs{E_k^\sigma(s^*)}\tau)$ and $\braket{\Gamma_k^\sigma}{0}=0$ for $1\leq k\leq n-1$.

Consider the measurement operator
\begin{equation}
\label{eqmain:measurement-state-space-markov}
\M=\Pi_X\otimes \ket{0}\bra{0} \otimes \ket{0}\bra{0},
\end{equation}
where $\Pi_X$ is a projection on the states of the Markov chain. Since $\braket{\Gamma_k^\pm}{0}=0$, the probability to measure $0$ in the third register is
\begin{equation}
    p_0=\abs{\alpha_n}^2+\sum_{\sigma=\pm 1}\sum_{k=1}^n\abs{\alpha_k^{\sigma}}^2\abs{\gamma_k^\sigma}^2.
\end{equation}
From $|\alpha_n|^2=\frac{1}{2}$ and $|\gamma_k^\sigma|\leq\pi/(\abs{E_k^\sigma(s^*)}\tau)$, we can bound $p_0$ as
\begin{align*}
    \frac{1}{2}\leq p_0&\leq \frac{1}{2}+\frac{\pi^2}{\tau^2}\sum_{\sigma=\pm 1}\sum_{k=1}^n\frac{\abs{\alpha_k^\sigma}^2}{\abs{E_k^\sigma(s^*)}^2}=\frac{1}{2}+\frac{\pi^2}{\tau^2}h(s^*)\\
    &\leq \frac{1}{2}+\frac{\pi^2}{\tau^2}HT(s^*)=\frac{1}{2}+\frac{1}{2}\delta^2,
\end{align*}
where we have used Prop.~\ref{prop-main:ht-connection} and the fact that $\tau=\frac{\pi}{\delta}\sqrt{2HT(s^*)}$.

In summary, we prepare with probability at least $1/2$ the state
$$
\ket{\psi_2}=
\frac{1}{\sqrt{p_0}}\left[
\alpha_n\ket{v_n(s^*),0}+\sum_{\sigma=\pm 1}\sum_{k=1}^{n-1}\alpha_k^\sigma\gamma_k^\sigma\ket{\Psi_k^{\sigma}(s^*)}
\right],
$$
where, from the calculation above, the second term has norm at most $\delta/\sqrt{2}$. This state is therefore close to $\ket{v_n(s^*),0}$, more precisely, we have
\begin{align*}
    \nrm{\ket{v_n(s^*),0}-\ket{\psi_2}}^2
    &=\left(1-\frac{\alpha_n}{\sqrt{p_0}}\right)^2+\frac{1}{p_0}\sum_{\sigma=\pm 1}\sum_{k=1}^n\abs{\alpha_k^{\sigma}}^2\abs{\gamma_k^\sigma}^2\\
    &\leq\left(1-\frac{1}{\sqrt{1+\delta^2}}\right)^2+\delta^2\\
    &\leq \delta^2+\delta^2\leq 2\delta^2.
\end{align*}

Let us finally consider the measurement in the state-space of the Markov chain. If the state we had prepared was $\ket{v_n(s^*),0}$ we would obtain a marked node with probability $1/2$ since $|\braket{v_n(s^*)}{M}|^2 = 1/2$. Equivalently, if we define the projector $\Pi_M$ to be the projector over marked nodes of the Markov chain, we have
$\nrm{\Pi_M\ket{v_n(s^*),0}}=\frac{1}{\sqrt{2}}$. The triangle inequality then implies
\begin{align*}
    \nrm{\Pi_M\ket{\psi_2}}
    &\geq\nrm{\Pi_M\ket{v_n(s^*),0}}-\nrm{\Pi_M\left(\ket{v_n(s^*),0}-\ket{\psi_2}\right)}\\
    &\geq\frac{1}{\sqrt{2}}-\sqrt{2}\delta.
\end{align*}
This implies that measuring $\ket{\psi_2}$ in the state-space of the Markov chain yields a marked node with probability at least $\nrm{\Pi_M\ket{\psi_2}}^2\geq\frac{1}{2}-2\delta$. Since, in the previous step we had prepared this state with probability at least $1/2$, overall the algorithm outputs a marked node with probability at least $1/4-\delta$.

Finally, Eq.~(\ref{eqmain:interpolated-vs-extended-hitting-time}) implies that for $s=s^*=1-p_M/(1-p_M)$, we have $HT(s^*)=\frac{HT^+(P,M)}{4}$, so that the time to execute Step 2 of the algorithm is $\tau=\frac{\pi}{\delta}\sqrt{2HT(s^*)}=\frac{\pi}{\delta}\sqrt{\frac{HT^+(P,M)}{2}}$, which concludes the proof.
\end{proof}
Algorithm \ref{algo:search-ctqw} can be thought of as an analog version of the quantum spatial search algorithm by Krovi et al \cite{krovi2016quantum}. Although this algorithm requires $\Theta(\log_2 HT^+(P,M))$ additional ancillary qubits unlike the algorithm of Ref.~\cite{chakraborty2018finding}, it will help create an intuitive understanding of our algorithm for solving $QSSamp$, discussed in Sec.~\ref{sec:preparation-of-stationary-state-von-neumann}. In fact, the cost of preparing the initial state of the spatial search algorithm corresponds to the $QSSamp$ problem.

In Ref.~\cite{chakraborty2018finding}, it was shown that $H(s)$ can be simulated be using only query access to the discrete-time quantum walk unitary $W(s)$ introduced in Ref.~\cite{krovi2016quantum}. This connection would allow us to quantify the running time of our quantum algorithm in terms of basic Markov chain operations.

To that end, given a Markov chain $P$, let us define the following oracular operations:
~\\
\begin{itemize}
    \item Check ($M$): Cost of checking whether a given node is marked. We denote this by $\C$.
    \item Update ($P$): Cost of applying one step of the walk $P$, which we denote by $\U$.
    \item Setup ($P$): The cost of preparing the initial state $\ket{v_n(0)}$, denoted by $\mathcal{S}$.
\end{itemize}
~\\
As from Refs.~\cite{krovi2016quantum,ambainis2019quadratic}, the cost of implementing $W(s)$, and consequently, the cost of evolving $H(s)$ for constant time, is in $\Oo(\C+\U)$, the running time of Algorithm \ref{algo:search-ctqw} is 

\begin{equation}
\label{eqmain:total-cost-search}
T=\Oo\left(\mathcal{S}+\sqrt{HT^+(P,M)}\left(\U+\C\right)\right)
\end{equation}

The $QSSamp$ problem helps quantify the cost $\mathcal{S}$ and intuitively, a quantum algorithm for this problem can be obtained by running the spatial search algorithm in reverse.
\section{Analog quantum algorithm to prepare coherent encoding of the stationary state of a Markov chain}
\label{sec:preparation-of-stationary-state-von-neumann}
In this section we describe our algorithm which, given a reversible Markov chain $P$ with stationary state $\pi=(\pi_1,\cdots,\pi_n)$, prepares a state that is $\epsilon$-close to the state 
\begin{equation}
\ket{\pi}=\sum_{x\in X}\sqrt{\pi_x}\ket{x}.
\end{equation}
A measurement in the basis spanned by the states of the Markov chain will allow us to sample from $\pi$, thereby solving the $QSSamp$ problem. From Fact \ref{fact:highest-eigenstate-equals-pi-s} and Proposition \ref{prop-main:zero-estate-somma-ortiz}, we have that 
\begin{equation}
\ket{\pi}=\ket{v_n(0)}.
\end{equation}
Thus, this is simply the highest eigenstate of the discriminant matrix $D(P)$ or equivalently, the $0$-eigenstate of $H(0)$. Therefore, given $P$, the problem of preparing $\ket{\pi}$ boils down to the state-generation problem just as in the case of spatial search.

Following Lemma \ref{lem_main:prep-state-von-neumann} and Corollary \ref{cor_main:prep-state-high-accuracy}, one can think of an algorithm to prepare $\ket{v_n(0)}$ as follows. 

Starting from some initial localized state $\ket{j,0}$ where $(j\in X)$, one can evolve according to the Hamiltonian $H(0)\otimes \widehat{p}$ for a time that scales as $\bigOt{1/\sqrt{\Delta}}$ to prepare $\ket{v_n(0)}$ with probability $|\braket{v_n(0)}{j}|^2\geq \eta$. Then by using $\Theta(1/\sqrt{\eta})$-rounds of (fixed-point) amplitude amplification \cite{yoder2014fixed}, one can prepare $\ket{v_n(0)}$. 

However, amplitude amplification is a discrete quantum algorithm and to the best of our knowledge it has no analog counterpart. As such, while constructing an analog quantum algorithm for this problem we cannot make use of amplitude amplification. We shall switch the value of $s$ to get around the need for amplitude amplification.

Consider the scenario where, given $P$, one marks a single state $j$, i.e.\ all the outgoing edges from $j$ are replaced with self-loops.  We denote the absorbing Markov chain corresponding to this $P'_j$. Then the resulting interpolated Markov chain is
\begin{equation}
\label{eqmain:interpolated-mc-single-marked}
P(s)=(1-s)P+sP_j'.
\end{equation} 
If the entry of the stationary state of $P$, corresponding to the marked element is $\pi_j$, then we find that $p_M=\pi_j$ and so for 
\begin{equation}
\label{eqmain:s*-mixing}
s=s^*=1-\pi_j/(1-\pi_j),
\end{equation} 
and from Eq.~\eqref{eqmain:highest-eigenstate-discriminant-matrix} we have that
\begin{align}
\label{eqmain:highest-estate-s*}
\ket{v_n(s^*)}&=\dfrac{1}{\sqrt{2}}\left(\dfrac{1}{\sqrt{1-\pi_j}}\sum_{x\neq j}\sqrt{\pi_x}\ket{x}+\ket{j}\right)\\
			  &=\dfrac{\ket{U}+\ket{j}}{\sqrt{2}}.
\end{align} 
Thus, the state $\ket{j}$ has a constant overlap with $\ket{v_n(s^*)}$. We can use this fact to  establish the following connection between the average hitting time and the eigenstates of $H(s^*)$, following Proposition \ref{prop-main:ht-connection}.

\begin{prop}
\label{prop-main:ht-interpolated-s*}
Let $h(s)$ be as defined in Proposition \ref{prop-main:ht-connection} and $HT(P,\{j\})$ be the average hitting time of $P$ with respect to the vertex $j$. Then, for $s^*=1-\pi_j/(1-\pi_j)$,
$$
h(s^*)=\sum_{k=1}^{n-1}\dfrac{|\braket{v_k(s^*)}{j}|^2}{1-\lambda^2_k(s^*)}.
$$
Furthermore,
\begin{equation}
HT(P,\{j\})/8\leq h(s^*)\leq HT(P,\{j\})/4.
\end{equation}
\end{prop} 
\begin{proof}
For $s=s^*$, we have that for any $1\leq k\leq n-1$, $\braket{v_k(s^*)}{v_n(s^*)}=0$ and substituting the expression of $\ket{v_n(s^*)}$, we obtain that for any $k$ such that $1\leq k\leq n-1$,
$$
|\braket{v_k(s^*)}{j}|^2=|\braket{v_k(s^*)}{U}|^2.
$$
So, we have:
$$
h(s^*)=\sum_{k=1}^{n-1}\dfrac{|\braket{v_k(s^*)}{U}|^2}{1-\lambda^2_k(s^*)}=\sum_{k=1}^{n-1}\dfrac{|\braket{v_k(s^*)}{j}|^2}{1-\lambda^2_k(s^*)}.
$$
Following proposition \ref{prop-main:ht-connection}, we have
$$
HT(s^*)/2\leq h(s^*)\leq HT(s^*).
$$
Since we have a single marked node $j$, the extended hitting time with respect to $j$ is equal to $HT(P,\{j\})$. The final expression of the proposition can be obtained by observing that from Eq.~\eqref{eqmain:interpolated-vs-extended-hitting-time}, we have $HT(s^*)=HT(P,\{j\})/4$.
\end{proof}
~\\
\RestyleAlgo{boxruled}
\begin{algorithm}[ht]
  \caption{Quantum algorithm to the prepare stationary state of any reversible Markov chain}\label{algo1}
  Let $\epsilon\in (0,1)$ and $\delta$ is a constant in $(0,1/4)$. 
\begin{itemize}
\item[1.~] Set $s=s^*=1-\pi_j/(1-\pi_j)$:
  \begin{itemize}
  \item[(a)] Evolve according to $H(s^*)\otimes\widehat{p}$ for time $T_1=\frac{\pi}{\delta}\sqrt{HT(P,\{j\})}$ starting from the state $\ket{j,0}\ket{x=0}$.
  \item[(b)] Post-select on obtaining $\ket{0}$ in the pointer register 
  \end{itemize}
Let the state obtained after step 1 be $\ket{\psi^{(1)}_f}$. 
\item[2.~] Reinitialize the pointer register.
\item[3.~] Set $s=0$:
\begin{itemize}
\item[(a)] Evolve $H(0)\otimes \widehat{p}$, starting from the state $\ket{\psi^{(1)}_f}\ket{x=0}^{\lceil\log_2(4/\epsilon)\rceil}$ for time $T_2=2\pi/\Delta(0)$.
\item[(b)] Repeat the Hamiltonian evolution in step (a) $\lceil\log_2(4/\epsilon)\rceil$ times, using a fresh block of $l=1+\lceil \log(1/\Delta(0))\rceil$-qubits in the \\ pointer register each time.
\item[(c)] Post-select on obtaining $\ket{0}$ in all $\lceil\log_2(4/\epsilon)\rceil$ copies of the pointer register. 
\end{itemize}
\item[4.~] Output the state of the first register.
\end{itemize} 
 \end{algorithm}
~\\
Recall that the initial state of Algorithm \ref{algo:search-ctqw} contained $\ket{v_n(0)}$ in the first register and our state-generation scheme resulted in the preparation of a state that has a constant overlap of $|\braket{v_n(s^*)}{v_n(0)}|=\Theta(1)$ with $\ket{v_n(s^*)}$. For our algorithm, we assume that for any $j\in X$, the state $\ket{j,0}$ is easy to prepare. The idea of the algorithm (See Algorithm \ref{algo1}) is to first invoke Lemma \ref{lem_main:prep-state-von-neumann} to prepare an intermediate state that has a constant overlap with $v_n(s^*)$. We will prove that this can be achieved by setting $s=s^*=1-\pi_j/(1-\pi_j)$, $\tau=\Theta(\sqrt{HT(P,\{j\})})$ and choosing the initial state to be $\ket{j,0}$. For the second stage, we set $s=0$ and start from the state obtained in stage 1. By invoking Corollary \ref{cor_main:prep-state-high-accuracy}, we prepare a state that is $\epsilon$-close to $\ket{v_n(0)}=\ket{\pi}$ with a constant probability. By this two-stage procedure, we can avoid the need to use amplitude amplification. We formally state the algorithm in Algorithm \ref{algo1} and prove its correctness in Lemma \ref{lem-main:algo1-proof}.
~\\
\begin{lemma}
\label{lem-main:algo1-proof}
Algorithm \ref{algo1} outputs a quantum state $\ket{\phi_f}$ such that
$$\|\ket{\phi_f}-\ket{\pi}\|\leq \epsilon,$$
in time
$$
T=\Theta\left(\sqrt{HT(P,\{j\})}\log(1/\epsilon)\right),
$$
with success probability at least $1/4-\delta$.
\end{lemma}
\begin{proof}  

First note that the 0-eigenstate of $H(s)$ is given by Eq.~\eqref{eqmain:highest-eigenstate-discriminant-matrix} and so for $s^*=1-\pi_j/(1-\pi_j)$, we have that 
\begin{equation}
\label{eqmain:highest-estate-s*}
\ket{v_n(s^*)}=\dfrac{1}{\sqrt{2}}\left(\dfrac{1}{\sqrt{1-\pi_j}}\sum_{x\neq j}\sqrt{\pi_x}\ket{x}+\ket{j}\right)
\end{equation}
Thus, the state $\ket{j,0}$ can be written in the eigenbasis of $H(s^*)$ as 
$$
\ket{j,0}=\alpha_n\ket{v_n(s^*),0}+\sum_{\sigma=\pm}\sum_{k=1}^{n-1}\alpha^{\sigma}_k\ket{\Psi^{\sigma}_k(s^*)},
$$
where $\alpha_n=1/\sqrt{2}$ and $\alpha^{\sigma}_k=\braket{\Psi^{\sigma}_k(s^*)}{j,0}$. 

Now, we shall evolve the state $\ket{j,0}\ket{x=0}$ according to the Hamiltonian $\tilde{H}(s^*)=H(s^*)\otimes \hat{p}$ for a time $\tau$. From Lemma \ref{lem_main:prep-state-von-neumann}, we have that the resulting state
\begin{align*}
\ket{\tilde{\psi}}&=\alpha_n\ket{v_n(s^*),0}\ket{0}\\
				  &+\sum_{\sigma=\pm 1}\sum_{k=1}^{n-1}\alpha_k^\sigma\ket{\Psi_k^{\sigma}(s^*)}\left(\gamma_k^\sigma\ket{0}+\Gamma_k^\sigma\ket{\Gamma_k^\sigma}\right),
\end{align*}
where $|\gamma_k^\sigma|\leq\pi/(\abs{E_k^\sigma(s^*)}\tau)$ and $\braket{\Gamma_k^\sigma}{0}=0$ for $1\leq k\leq n-1$.

The probability to measure $\ket{0}$ in the third register is
\begin{equation}
    p_0=\abs{\alpha_n}^2+\sum_{\sigma=\pm 1}\sum_{k=1}^n\abs{\alpha_k^{\sigma}}^2\abs{\gamma_k^\sigma}^2.
\end{equation}  
From $|\alpha_n|^2=\frac{1}{2}$, $|\gamma_k^\sigma|\leq\pi/(\abs{E_k^\sigma(s^*)}\tau)$ and the analysis in Lemma \ref{lem_main:correctness-search-ctqw}, we can bound $p_0$ as
\begin{align*}
    \frac{1}{2}\leq p_0&\leq \frac{1}{2}+\frac{\pi^2}{\tau^2}h(s^*)\\
    &\leq \frac{1}{2}+\frac{\pi^2}{4\tau^2}HT(P,\{j\})=\frac{1}{2}+\frac{1}{2}\delta^2,
\end{align*}
where we have used Prop.~\ref{prop-main:ht-interpolated-s*} and the fact that $\tau=\frac{\pi}{\delta}\sqrt{HT(P,\{j\})/2}$. Thus,  after stage one, with probability $p_0\geq 1/2$, we have prepared the state
\begin{align}
\ket{\psi^{(1)}_f}&=\frac{1}{\sqrt{p_0}}\left[\alpha_n\ket{v_n(s^*),0}+\sum_{\sigma=\pm 1}\sum_{k=1}^{n-1}\alpha_k^\sigma\gamma_k^\sigma\ket{\Psi_k^{\sigma}(s^*)}\right]\nonumber\\
&=\ket{v_n(s^*),0}+\ket{\text{err}}\label{eqmain:intermediate-state}
\end{align}
in the first register, where $\spnorm{\ket{\text{err}}}\leq\sqrt{2}\delta$. The running time of stage one of the algorithm is
$$
T_1=\frac{\pi}{\delta}\sqrt{HT(P,\{j\})/2}.
$$
For any constant $\delta\in (0,1)$, we have $T_1=\Theta(\sqrt{HT(P,\{j\}})$. 

We shall use this state as the initial state for stage two of our algorithm. In this stage, we set $s=0$ and prepare a state that is $\epsilon$-close to $\ket{v_n(0)}$. Observe that the overlap $\braket{v_n(0)}{v_n(s^*)}=1/\sqrt{2}$, so
\begin{align*}
\abs{\braket{v_n(0),0}{\psi^{(1)}_f}}
&=\abs{\braket{v_n(0),0}{v_n(s^*),0}+\braket{v_n(0),0}{\text{err}}}\\
&\geq\frac{1-2\delta}{\sqrt{2}}.
\end{align*}

We begin by expressing $\ket{\psi^{(1)}_f}$ in the eigenbasis of $H$. We have
$$
\ket{\psi^{(1)}_f}=\delta_n\ket{v_n(0),0}+\sum_{\sigma=\pm 1}\sum_{k=1}^{n-1}\delta_k^\sigma\ket{\Psi_k^{\sigma}(0)},
$$ 
where $\abs{\delta_n}\geq (1-2\delta)/\sqrt{2}$ and for any $1\leq k \leq n-1$, $\delta_k^\sigma=\braket{\Psi_k^{\sigma}(0)}{\psi^{(1)}_f}$. 

Let $\Delta(0)=\sqrt{1-\lambda^2_{n-1}(0)}$, be the eigenvalue gap between the $0$-eigenstate of $H(0)$ and the rest of the spectrum. So, if $\Delta=1-\lambda_{n-1}(0)$ is the spectral gap of $P$, $\sqrt{\Delta}\leq \Delta(0)\leq \sqrt{2\Delta}$.

From Lemma \ref{lem_main:prep-state-von-neumann}, by choosing $\tau=2\pi/\Delta(0)$, we obtain a state
\begin{align*}
\ket{\tilde{\psi}}&=\delta_n\ket{v_n(0),0}+\\
			&=\sum_{\sigma=\pm 1}\sum_{k=1}^{n-1}\delta_k^\sigma\ket{\Psi_k^{\sigma}(0)}    \left(\gamma_k^\sigma\ket{0}+\Gamma_k^\sigma\ket{\Gamma_k^\sigma}\right),
\end{align*}
where now, for any $k$ such that $1\leq k\leq n-1$, $|\gamma_k^\sigma|\leq1/2$, $|\Gamma_k|\geq\sqrt{3}/2$ and $\braket{\Gamma_k}{0}=0$.

We can now use Corollary \ref{cor_main:prep-state-high-accuracy}, with $l=1+\lceil \log_2 (1/\Delta(0))\rceil$, $\epsilon'=\epsilon/4\leq|\delta_n|\epsilon/\sqrt{2}$ and $m=l\cdot\lceil\log(1/\epsilon')\rceil$. By choosing $\tau=2\pi/\Delta(0)$, we prepare a quantum state $\ket{\phi_f}$ that is $\epsilon$-close to $\ket{v_n(0)}$ with probability at least $\abs{\delta_n}^2\geq 1/2-2\delta$. The running time of the second stage of the algorithm is in 
$$
T_2=\dfrac{2\pi}{\Delta(0)}.\lceil\log(1/\epsilon')\rceil=\Theta\left(\dfrac{1}{\sqrt{\Delta}}\log(1/\epsilon)\right).
$$

Now putting things together, Algorithm \ref{algo1} prepares a state $\ket{\phi_f}$ which is $\epsilon$-close to $\ket{v_n(0)}$. The first stage succeeded with probability $p_0\geq 1/2$ and so the overall success probability is $\geq 1/4-\delta$. The overall running time
$$
T=T_1+T_2=\Theta\left(\sqrt{HT(P,\{j\})}+\dfrac{1}{\sqrt{\Delta}}\log(1/\epsilon)\right).
$$
Interestingly, by Eq.~(\ref{eqmain:mixing-time-definition}) this also implies that the running time of our algorithm is actually the sum of the square root of the classical hitting time and the square root of the classical mixing time, i.e.\
$$T=\widetilde{\Theta}\left(\sqrt{HT(P,\{j\})}+\sqrt{T_{\mathrm{mix}}}\right).$$ 
Note that in general, the hitting time is at least as large as the mixing time of an ergodic, reversible Markov chain. Thus the running time is in fact, 
\begin{equation}
\label{eqmain:run-time-qsamp}
T=\Theta\left(\sqrt{HT(P,\{j\})}\log(1/\epsilon)\right).
\end{equation}
\end{proof}
~\\~\\
As mentioned in the previous section, the $QSSamp$ problem helps quantify the Setup cost $\mathcal{S}$ of the spatial search problem (See Eq.~\eqref{eqmain:total-cost-search}). As such, Algorithm \ref{algo1} implies that the setup cost of Algorithm \ref{algo:search-ctqw} is given by Eq.~\eqref{eqmain:run-time-qsamp}.
\section{Time-averaged quantum mixing:~Limiting distribution and mixing time}
\label{sec-main:mixing-time-erg}
Now we shall deal with the $QLSamp$ problem and the notion of mixing time that arises from this problem. For any ergodic, reversible Markov chain $P$, we have seen from Sec.~\ref{subsec:quantities-related-to-mc} that it is possible to sample from its distribution at $T\rightarrow\infty$ (limiting distribution) after a time $T_{\mathrm{mix}}=\bigOt{1/\Delta}$, known as the mixing time of $P$, where $\Delta$ is the spectral gap of $P$. In fact, any initial distribution converges to the stationary distribution $\pi$ after $T_{\mathrm{mix}}$ applications of $P$. In a strict sense, such a limiting distribution is absent for quantum walks as the underlying dynamics is unitary and hence, distance preserving. 

However, one can define the mixing of quantum walks on a graph is defined in a time-averaged sense: the probability that the walker is at some node $f$ after some time $t$, picked uniformly at random in the interval $[0,T]$ \cite{aharonov2001quantum}. This gives a \textit{time-averaged probability distribution} at any time $t$ and also a \textit{limiting probability distribution} as $T\rightarrow\infty$. The mixing time of a quantum walk on any ergodic, reversible Markov chain $P$ is the time after which the time-averaged probability distribution is $\epsilon$-close to the limiting probability distribution.  
%

Consider any ergodic, reversible Markov chain $P$ with $|X|=n$. Given $P$, suppose $H_P$ denotes the underlying Hamiltonian corresponding to a quantum walk on $P$. We require that the eigenvalues of $H_P$ lie between $-1$ and $1$, i.e.\ $\nrm{H_P}=1$. Let the spectral decomposition of $H_P=\sum_i\lambda_i\ket{v_i}\bra{v_i}$ where $\ket{v_i}$ is the eigenstate corresponding to the eigenvalue $\lambda_i$, $i\in\{1,2,...,n\}$. Furthermore, suppose that the initial state of the walker is $\ket{\psi_0}$.

Consequently, the state of the walker after a time $t$ is governed by the Schr\"odinger equation, i.e.
\begin{equation}
\label{eqmain:state_time_t}
\ket{\psi(t)}=e^{-iH_Pt}\ket{\psi_0}.
\end{equation}
In order to define a limiting distribution for quantum walks, one obtains a C\'{e}saro-average of the probability distribution, i.e.\ one evolves for a time $t$ chosen uniformly at random between $0$ and $T$  followed by a measurement. The average probability that the state of the walker is some localized node $\ket{f}$ is given by 
\begin{equation}
\label{eqmain:prob_t}
P_{f}(T)=\frac{1}{T}\int_{0}^{T} dt\ |\braket{f|e^{-iH_Pt}}{\psi_0}|^2.
\end{equation}
Thus, as $T\rightarrow\infty$, this leads to a limiting probability distribution, i.e.\  
\begin{equation}
\label{eqmain:prob_infinite}
P_{f}(T\rightarrow\infty)=\lim_{T\rightarrow\infty}P_{f}(T)=\sum_{\lambda_i=\lambda_l}\braket{v_l}{f}\braket{f}{v_i}\braket{v_i}{\psi_0}\braket{\psi_0}{v_l},
\end{equation}
where the sum is over all pairs of degenerate eigenvalues. So, if $H_P$ has a simple spectrum, i.e.\ all its eigenvalues are distinct, then the sum is over all its eigenvalues. 

In order to calculate how fast the instantaneous time-averaged distribution of the quantum walk converges to this limiting distribution, i.e.\ we need to bound the quantity $
\nrm{P_{f}(T\rightarrow\infty)-P_{f}(T)}_1$. 

In fact, it is easy to verify that they are $\epsilon$-close, i.e.\
$$\nrm{P_{f}(T\rightarrow\infty)-P_{f}(T)}_1\leq \epsilon,$$
as long as
\begin{align}
T=\Omega\left(\dfrac{1}{\epsilon}\sum_{\lambda_i\neq \lambda_l}\dfrac{\left |\braket{v_i}{\psi_0}\right |. \left | \braket{\psi_0}{v_{l}}\right |}{\left|\lambda_{l}-\lambda_i\right|}\right).
\end{align}

This naturally leads to the following upper bound on the quantum mixing time
\begin{equation}
\label{eqmain:mixing-time-upper-bound}
T_{\mathrm{mix}}=\Oo\left(\dfrac{1}{\epsilon}\sum_{\lambda_i\neq \lambda_l}\dfrac{\left |\braket{v_i}{\psi_0}\right |. \left | \braket{\psi_0}{v_{l}}\right |}{\left|\lambda_{l}-\lambda_i\right|}\right).
\end{equation} 

There do exist differences between the quantum and classical limiting distributions. For example, in the quantum case, the limiting distribution is dependent on the initial state of the quantum walk. Also, unlike classical random walks, the quantum mixing time depends on all the eigenvalue gaps of $H_P$ as opposed to only the spectral gap. 

Ignoring the numerator in the right hand side of Eq.~\eqref{eqmain:mixing-time-upper-bound}, we need to evaluate the following quantity in order to upper-bound $T_{\mathrm{mix}}$:	
\begin{equation}
\label{eqmain:double-sum}
\Sigma=\sum_{\lambda_i\neq\lambda_l}\dfrac{1}{\left|\lambda_{l}-\lambda_i\right|}.
\end{equation}

As such we intend to obtain the best possible bounds for this quantity. To this end, let us define $\Delta_{\min}$ as the minimum eigenvalue gap of $H_P$, over all pairs of distinct eigenvalues, i.e.\
\begin{equation}
\label{eqmain:min-eigenvalue-gap-definition}
\Delta_{\min}=\min_{\\ i,j, \lambda_i\neq\lambda_j}\left\{|\lambda_i-\lambda_j|, s.t.~i\neq j \right\}.
\end{equation}
Note that this is different from the spectral gap $\Delta$, which is the difference between the two highest eigenvalues of $H_P$. We prove the following:
\\~\\
\begin{lemma}
\label{lem:bounds-sigma}
If $\Sigma$ and $\Delta_{\min}$ are defined as in Eqs.~\eqref{eqmain:double-sum} and \eqref{eqmain:min-eigenvalue-gap-definition} then
\begin{equation}
\label{eqmain:double-sum-bounds}
\dfrac{1}{\Delta_{\min}}\leq \Sigma \leq \bigOt{\dfrac{n}{\Delta_{\min}}}
\end{equation}
~\\ 
\end{lemma}
\begin{proof}
The lower bound is straightforward by noting that $\exists i,l$ such that $|\lambda_{l}-\lambda_i|=\Delta_{\min}$.

For the upper bound, we have that 
\begin{equation}
\Sigma \leq \dfrac{1}{\Delta_{\min}}\sum_{l\neq i}\dfrac{1}{|\lambda_{l}-\lambda_i|},
\end{equation}
where we have used the fact that for any $\lambda_i\neq \lambda_l$, $|\lambda_l-\lambda_i|\geq |l-i|\Delta_{\min}$.
This implies that if $|l-i|=r$, we obtain 
\begin{align}
\Sigma &\leq \dfrac{1}{\Delta_{\min}}\sum_{r=1}^{n-1}\sum_{l,i:|l-i|=r}\dfrac{1}{r}\\
       &\leq \dfrac{1}{\Delta_{\min}} (n-1)\left(1+\dfrac{1}{2}+\cdots+\dfrac{1}{n-1}\right)\\
       &\leq \dfrac{n\log n}{\Delta_{\min}}=\bigOt{\dfrac{n}{\Delta_{\min}}}.
\end{align}
\end{proof}
~\\
The upper bound on $\Sigma$ obtained in Lemma \ref{lem:bounds-sigma} leads directly to an upper bound on the quantum mixing time $T_{\mathrm{mix}}$. This can be seen from the fact that the numerator in Eq.~\eqref{eqmain:mixing-time-upper-bound} is less than one and so 
\begin{equation}
\label{eqmain:mixing-time-generic-ub}
T_{\mathrm{mix}}=\bigOt{n/\Delta_{\min}},
\end{equation}
for any $H_P$.

All prior works hitherto have analyzed the $QLSamp$ problem for simple, unweighted graphs. Given a graph $G(V,E)$ of $|V|=n$ nodes and $|E|$ edges, the underlying quantum walk is defined on the nodes of the graph with corresponding Hamiltonian being the (normalized) adjacency of the graph. That is, $H_P=A_G/\nrm{A_G}$, where $A_G$ is an $n\times n$ matrix such that each entry
\begin{equation}
a_{ij}=
\begin{cases}
&1,~~~(i,j)\in E\\
&0,~~~\text{otherwise},
\end{cases}
\end{equation}
and $\nrm{A_G}$ is the spectral norm of $A_G$.
  
In this section, we elaborate on the results of Ref.~\cite{chakraborty2020fast} and provide numerical evidence to back up our analytical bounds. In particular, we focus on the random matrix theory aspects of our proof, elaborating on the underlying concepts. Finally, by defining a quantum walk on the edges as in Sec.~\ref{subsec:hamiltonian-definition}, we extend the notion of $QLSamp$ to any ergodic, reversible Markov chain.
~\\
\subsection{Erd\"os-Renyi random graphs} 
Let us consider a graph $G$ with a set of vertices $V=\{1,\dots,n\}$. We restrict ourselves to simple graphs, i.e.\ unweighted graphs which do not contain self-loops or multiple edges connecting the same pair of vertices. The maximum number of edges that a simple graph $G$ can have is $N={n \choose 2}$. Thus, there are ${N \choose M}$ graphs of $M$ edges and the total number of (labelled) graphs is $\sum_{M=0}^N {N \choose M}=2^N$ \cite{graph_enumeration}. We consider the random graph model $G(n,p)$, a graph with $n$ vertices where we have an edge between any two vertices with probability $p$, independently of all the other edges \cite{ER59,ER60,bollobas_book} (See Fig.~\ref{figmain:erg-30-nodes-p-0.2}). In this model, a graph $G_0$ with $M$ edges appears with probability $P\{G(n,p)=G_0\}=p^M (1-p)^{N-M}$. In particular, if we consider the case $p=1/2$, each of the $2^N$ graphs appears with equal probability $P=2^{-N}$. We shall refer to random graphs having a constant $p$ as a \textit{dense random graph}. On the other hand, random graphs for which $p=o(1)$, i.e.\ when $p$ decreases with $n$ shall be referred to as \textit{sparse random graphs}. 

\begin{figure}[h!]
\includegraphics[scale=0.22]{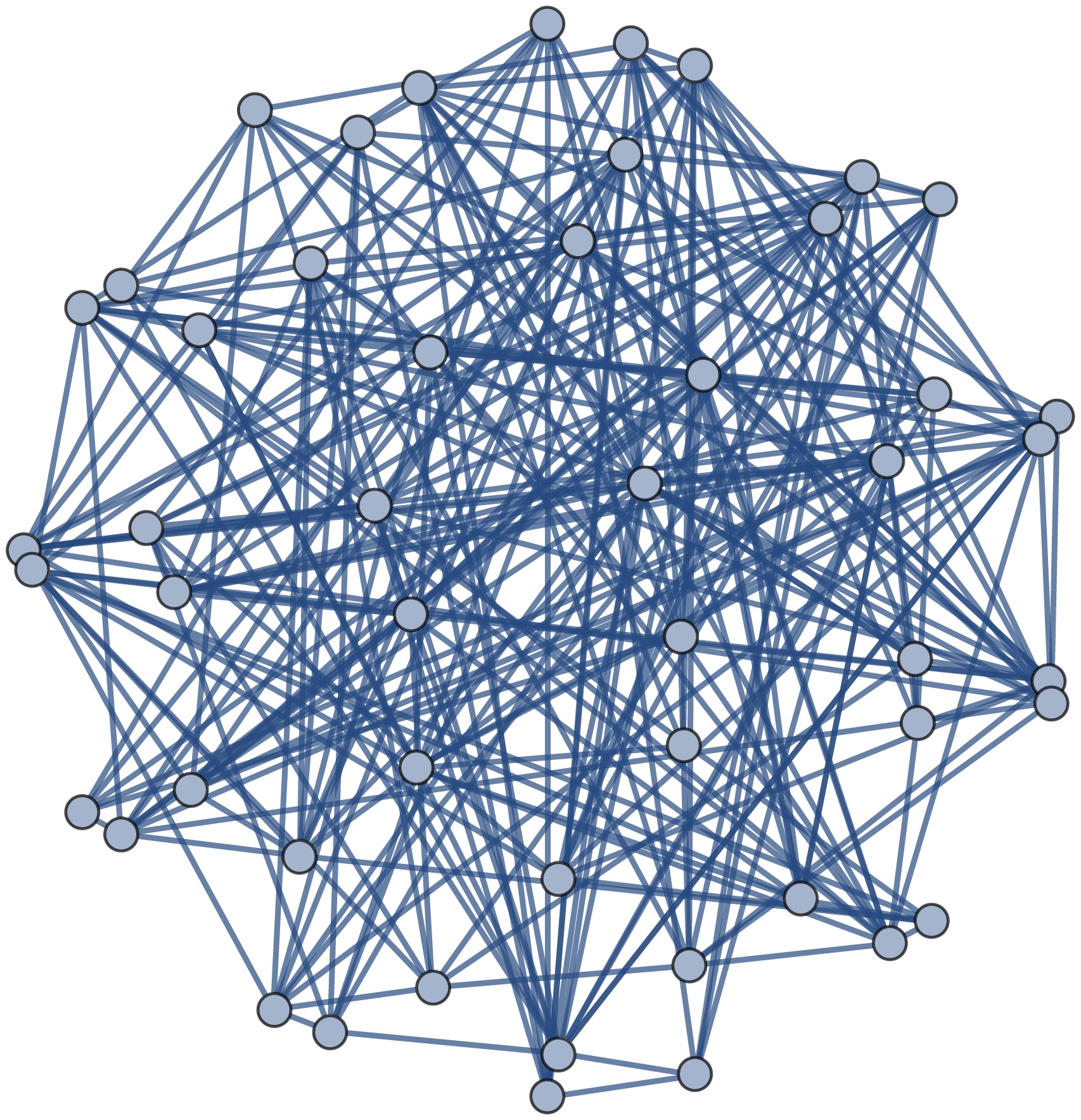}
\caption{\footnotesize{Erd\"os-Renyi random graph $G(50,0.2)$.}}
\label{figmain:erg-30-nodes-p-0.2}
\end{figure}

In their seminal papers, Erd\"{o}s and R\'enyi introduced this model of random graphs and studied the probability of a random graph to possess a certain property $Q$ \cite{ER59,ER60}. For example, they investigated properties such as the connectedness of the graph, the probability that a certain subgraph is present, etc. They stated that \emph{almost all graphs} have a property $Q$ if the probability that a random graph $G(n,p)$ has $Q$ goes to $1$ as $n\rightarrow \infty$. Equivalently, it can be stated that $G(n,p)$ \emph{almost surely} has property $Q$, i.e.\ property $Q$ holds with probability $1-o(1)$. 

Interestingly, certain properties of random graphs arise suddenly for a certain critical probability $p=p_c$, where this probability depends typically on $n$. More precisely, if $p(n)$ grows faster than $p_c(n)$, the probability that the random graph has property $Q$
goes to $1$ in the asymptotic limit, whereas if it grows slower than $p_c(n)$ it goes to $0$. For example when $p>\log(n)/n$ the graph is almost surely connected, whereas if $p<\log(n)/n$ the graph has almost surely isolated nodes. 

Here we shall concern ourselves with random graphs above the percolation threshold and calculate an upper bound on the quantum mixing time for quantum walks on such graphs. Observe that for a random graph, $G(n,p)$, its adjacency matrix, which we denote as $A_{G(n,p)}$, is an $n\times n$ symmetric matrix with each non-diagonal entry being $1$ with probability $p$ and $0$ with probability $1-p$. All diagonal entries of $A_{G(n,p)}$ are $0$. Thus $A_{G(n,p)}$ is a discrete random matrix and knowledge of its eigenvalues and eigenvectors is crucial to obtaining the quantum mixing time. 

Finally, from the aforementioned discussion, obtaining the quantum mixing time for $G(n,p)$ can be interpreted as holding for \textit{almost all} graphs.     
~\\
\subsection{Random matrices: Spectral properties of $\mathbf{A_{G(n,p)}}$} 

Here we look at the eigenvalues and eigenvectors of the random matrix $A_{G(n,p)}$. 

As mentioned earlier, the Hamiltonian corresponding to the quantum walk on $G(n,p)$ is simply the normalized Adjacency matrix of $G(n,p)$. The highest eigenvalue of $A_{G(n,p)}$, converges to a Gaussian distribution with mean $np$ and standard deviation $\sqrt{p(1-p)}$, as $n\rightarrow \infty$. This fact was first shown in Ref.~\cite{furedi1981eigenvalues} for constant $p$ and was later improved for sparse random graphs ($p=o(1)$) in Ref.~\cite{erdHos2013spectral}. In fact, as we shall show shortly it suffices to consider the matrix 
\begin{equation}
\label{eqmain:normalized-adjacency-matrix-erg}
\bar{A}_{G(n,p)}=\dfrac{A_{G(n,p)}}{np},
\end{equation}
as the quantum walk Hamiltonian. 

Let the eigenvalues of $\bar{A}_{G(n,p)}$ be $\lambda_n >\lambda_{n-1}\geq \cdots \lambda_1$, such that $\ket{v_i}$ is the eigenvector corresponding to the eigenvalue $\lambda_i$, $i\in\{1,2,...,n\}$, i.e.\ $\bar{A}_{G(n,p)}\ket{v_i}=\lambda_i\ket{v_i}$. Then we have that for $p\geq \log^8(n)/n$,
\begin{equation}
\label{eqmain:highest-eigenvalue-erdos-renyi}
\lambda_n=1+\sqrt{\dfrac{1-p}{np}}o(1)+o\left(\dfrac{1}{n\sqrt{p}}\right),
\end{equation}
with probability $1-o(1)$, which implies that $\nrm{\bar{A}_{G(n,p)}}\approx 1$ \cite{erdHos2013spectral, chakraborty2020fast}.

It can also be shown that for the same range of $p$, the second highest eigenvalue $\lambda_{n-1}$ can be upper bounded as
\begin{equation}
\label{eqmain:second-highest-evalue-erg}
\lambda_{n-1}\leq\dfrac{6}{\sqrt{np}}+\mathcal{O}\left(\dfrac{\log(n)}{(np)^{3/4}}\right),
\end{equation} 
with probability $1-o(1)$ \cite{furedi1981eigenvalues, vu2007, chakraborty2020fast}. This immediately implies that the spectral gap of $\bar{A}_{G(n,p)}$, $\Delta=\bigOt{1}$. Consequently, a classical random walk on $G(n,p)$ mixes quite fast - in $\bigOt{1}$ time. 

However, it is clear from the expression for $T^{G(n,p)}_{\mathrm{mix}}$ in Eq.~\eqref{eqmain:mixing-time-upper-bound} that the knowledge of \textit{all} eigenvalue gaps are crucial in obtaining the quantum mixing time. As such require the knowledge of the spacings between all the eigenvalues of the random matrix $\bar{A}_{G(n,p)}$.
~\\~\\
\textbf{Semicircle law:~} It is well known that as $np\rightarrow\infty$, the spectral density of the bulk of the spectrum of $A_{G(n,p)}$ converges to the well known semicircle distribution given by
\begin{equation}
\label{eqmain:semicircle_law}
\rho_{sc}(\lambda) = \begin{cases} \dfrac{\sqrt{4 n p(1-p)-\lambda^2}}{2\pi n p(1-p)} &\mbox{if } |\lambda|<2\sqrt{np(1-p)} \\
0 & \mbox{otherwise }\end{cases}.
\end{equation} 
This implies that $\Theta(n)$ eigenvalues of $A_{G(n,p)}$ lie within $[-R,R]$ where
\begin{equation}
R=2\sqrt{np(1-p)},
\end{equation} 
is the radius of the semicircle. On applying the appropriate normalization, we find that the spectral density of $\bar{A}_{G(n,p)}$ converges to a semicircle of radius
$$
\bar{R}=2\sqrt{\dfrac{1-p}{np}}.
$$
The fraction of eigenvalues of $\bar{A}_{G(n,p)}$ lying in some spectral window $\mathcal{I}\in [-R,R]$ converges to the area of the semicircle within $\mathcal{I}$ as $np\rightarrow\infty$. However, the semicircle law provides only a macroscopic description of the eigenvalues of $\bar{A}_{G(n,p)}$, i.e.\ the aforementioned result holds only when $|\mathcal{I}|\gg 1$. However, in order to obtain the quantum mixing time, we need information about all eigenvalue gaps including consecutive gaps where \ $|\I|\sim 1/n$, which renders the semicircle law useless. As a result, for our purposes we need to look at mesoscopic and microscopic statistics of eigenvalues of $\bar{A}_{G(n,p)}$.

However, for subsequent analysis, we shall require two results that can be obtained from the semicircle law itself which we state now. Note that there are $\Theta(n)$ eigenvalues with a radius of $\bar{R}$. This directly gives information about the average eigenvalue gap of $\bar{A}_{G(n,p)}$ given by
\begin{equation}
\label{eqmain:avg-eigenvalue-gap-erg}
\bar{\Delta}=\Theta\left(\dfrac{1}{n^{3/2}\sqrt{p}}\right).
\end{equation}
Also from the semicircle law itself, one can define the so-called \textit{classical eigenvalue locations} of $\bar{A}_{G(n,p)}$. For each $1\leq i\leq n-1$, we can define the classical location $\gamma_i$ as the solution to the following equation
\begin{equation}
\int_{-\infty}^{\gamma_i} \rho_{sc}(x) \,dx = \frac{i}{n}.
\end{equation}
Thus, the position of the $i^{\mathrm{th}}$ classical location is obtained by filling up $i/n$-area of the semicircle. From this condition one obtains that for $i\leq n/2$, $r\leq n-2i$ and some universal constant $c>0$, 
\begin{equation}
\label{eqmain:classical-location-distance}
\gamma_{i+r}-\gamma_i \geq c\dfrac{r}{n^{7/6}i^{1/3}\sqrt{p}}. 
\end{equation}
An identical estimate holds for the other half of the spectrum by symmetry. 
\\~\\
\textbf{Eigenvalue rigidity criterion:~} The semicircle law was shown to hold for smaller spectral windows in Refs.~\cite{erdHos2012rigidity,erdHos2013spectral}. An immediate consequence of this fact is that the every eigenvalue (with the exception of $\lambda_n$) of $\bar{A}_{G(n,p)}$ is located close to their classical eigenvalue positions. Formally, they showed that for $n^{-1/3}\leq p \leq 1-n^{-1/3}$ and any $\eps\geq 0$ the eigenvalues of $bar{A}_{G(n,p)}$ satisfy
\begin{equation}
\label{eqmain:eigenvalue-rigidity}
|\lambda_i - \gamma_i| \leq \frac{n^{\eps} (n^{-2/3} \alpha_i^{-1/3} + n^{-1-\phi} )}{(pn)^{1/2}}
\end{equation}
with probability $1 - o(1)$, where
$$
\phi := \frac{\log p}{ \log n} \quad \text{and} \quad \alpha_i := \max\{i, n-i\}.
$$
\begin{figure}[h!]
\centering
\includegraphics[scale=0.4]{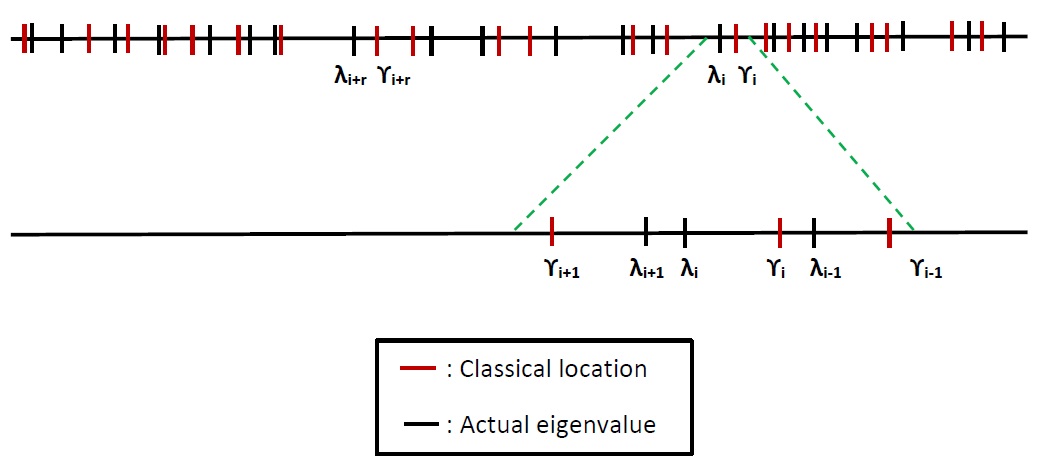}
\caption{A pictorial representation of eigenvalue rigidity: The actual eigenvalue locations of adjacency matrices of Erd\"os-Renyi random graphs, $\lambda_i$ (denoted by solid strokes) are close to their classical eigenvalue positions, $\gamma_i$ (denoted by dashed strokes) as predicted by the semicircle law. Although eigenvalue rigidity provides information about eigenvalue value gaps that are far away, it does not provide any information about the smallest eigenvalue gaps (such as consecutive eigenvalue gaps). As a result the eigenvalue rigidity criterion provides information about eigenvalues at a mesoscopic scale and information about the smallest gaps are obtained from eigenvalue statistics at a microscopic scale.}
\label{fig:rigidity-vs-gaps}
\end{figure}
Eigenvalue rigidity does reveal information about eigenvalue gaps of $\bar{A}_{G(n,p)}$. Note that for any $r\geq 1$ one obtains
\begin{equation}
\lambda_{i+r}-\lambda_{r}=\left(\lambda_{i+r}-\gamma_{i+r}\right)+\left(\gamma_{r}-\lambda_r\right)+\left(\gamma_{i+r}-\gamma_{r}\right).
\end{equation}
As a result, whenever $|\gamma_{i+r}-\gamma_{r}|$ scales larger than $\left|\lambda_{i+r}-\gamma_{i+r}\right|+\left|\lambda_r-\gamma_{r}\right|$, eigenvalue rigidity kicks in and an accurate estimate of $\lambda_{i+r}-\lambda_i$ is given by the difference between their classical eigenvalue locations, $\gamma_{i+r}-\gamma_i$. That is, there exists some $r=r^\star(i)$, such that for all $r\geq r^\star(i)$,
$$
\lambda_{i+r}-\lambda_{r}\approx \left(\gamma_{i+r}-\gamma_{r}\right). 
$$
From Eq.~\eqref{eqmain:classical-location-distance} and Eq.~\eqref{eqmain:eigenvalue-rigidity}, we obtain that 
\begin{equation}
\label{eqmain:rigidity-dominates}
r^\star(i)=n^{\eps} \max\{1, n^{2/3} \alpha_i^{1/3} n^{-1-\phi}\} \leq n^{\eps-\log p/\log n}.
\end{equation}
As such we cannot exploit eigenvalue rigidity to estimate gaps of the form $|\lambda_l-\lambda_i|$ as long as $|l-i|\leq r^\star(i)$. Thus, eigenvalue rigidity does not provide information about the smallest eigenvalue gaps (for a pictorial representation of this fact, see Fig.~\ref{fig:rigidity-vs-gaps}) and reveals eigenvalue statistics of $\bar{A}_{G(n,p)}$ at a mesoscopic scale. Thus in order to obtain information about consecutive eigenvalue gaps of $\bar{A}_{G(n,p)}$, we shall need to go to a microscopic scale.
\\~\\
\textbf{Microscopic eigenvalue statistics of $\mathbf{\bar{A}_{G(n,p)}}$:} At the microscopic scale, results are notoriously difficult to obtain. Tao and Vu \cite{tao2014random} showed that $\bar{A}_{G(n,p)}$ has a simple spectrum for dense graphs, resolving a long-standing conjecture due to Babai \cite{babai1982isomorphism}. Recently, this was resolved also for sparse graphs \cite{luh2018sparse}. We state their results formally.
\\~\\
\begin{fact}
\label{fact:erg-simple-spectrum}
There exists a constant $C>0$ such that for $\frac{C\log^6(n)}{n}\leq p \leq 1-\frac{C\log^6(n)}{n}$, $\bar{A}_{G(n,p)}$ has a simple spectrum with probability $1-o(1)$.
\end{fact}
~\\
The fact that every eigenvalue gap is non-zero implies that the expression for $\bar{A}_{G(n,p)}$, the double sum $\Sigma$ in Eq.~\eqref{eqmain:double-sum} can now be re-written as
\begin{align}
\label{eqmain:double-sum-defintion-1}
\Sigma&=\underbrace{\sum_{i=1}^{n-1}\dfrac{1}{\left|\lambda_{i+1}-\lambda_i\right|}}_{\Sigma_1}+\underbrace{\sum_{i=1}^{n-2}\dfrac{1}{\left|\lambda_{i+2}-\lambda_i\right|}}_{\Sigma_2}+\cdots\\
\label{eqmain:double-sum-defintion-2}
      &=\sum_{r=1}^{n-1}\Sigma_r=\sum_{r=1}^{n-1}\sum_{i=1}^{n-r}\dfrac{1}{\left|\lambda_{i+r}-\lambda_i\right|},
\end{align} 
while the limiting probability distribution is
\begin{equation}
\label{eqmain:prob_infinite}
P_{f}(T\rightarrow\infty)=\lim_{T\rightarrow\infty} P_{f}(T)=\sum_{i=1}^n|\braket{f}{v_i}\braket{v_i}{\psi_0}|^2.
\end{equation}
Furthermore, Nguyen, Tao and Vu \cite{nguyen2015random} proved that all the eigenvalue gaps of $\bar{A}_{G(n,p)}$ for dense random graphs are not only non-zero, but also separated. This was improved for the case of sparse random graphs by Lopatto and Luh \cite{lopatto2019tail}. In fact they asked the following question: how likely is it for any eigenvalue gap $\delta_i=\lambda_{i+1}-\lambda_i$ to be less than some $\delta$ times the average gap $\bar{\Delta}$? They proved that there exists a constant $C>0$ such that for $n^{-1/3}\leq p \leq 1-n^{-1/3}$,
\begin{equation}
\label{eqmain:tail-bound-random-graphs} 
\sup_{1 \leq i \leq n-1} \mathbb{P} \left( \delta_i \leq \frac{\delta}{n^{3/2} \sqrt{p}} \right) \leq C \delta \log n,
\end{equation}
for all $\delta\geq n^{-C}$.
~\\
Applying a union bound to this gives a lower bound on the minimum eigenvalue gap $\Delta_{\min}$ (defined in Eq.~\eqref{eqmain:min-eigenvalue-gap-definition}) for $\bar{A}_{G(n,p)}$. We prove that here.
~\\
\begin{lemma}{Lower bound on $\Delta_{\min}$:~} For $p\geq n^{-1/3}$,
\begin{equation}
\label{eqmain:lower-bound-min-gap}
\Delta_{\min}\geq \dfrac{1}{n^{5/2+o(1)}\sqrt{p}}.
\end{equation}
with probability $1-o(1)$.
\end{lemma}
~\\
\begin{proof}
Let $\A_i$ be the event that $\delta_i\leq \frac{\delta}{n^{3/2} \sqrt{p}}$. Then using the union bound and Eq.~\eqref{eqmain:tail-bound-random-graphs}, we obtain
\begin{equation}
\label{eqmain:union-bound}
\mathbb{P}\left(\bigcup_{i}\A_i\right)\leq \sum_{i}\mathbb{P}\left(\A_i\right)\leq C\ n\ \delta\log n.
\end{equation}
This implies that the probability that at least one of the gaps is less than $\frac{\delta}{n^{3/2} \sqrt{p}}$ is upper bounded by the right hand side of Eq.~\eqref{eqmain:union-bound}. By choosing
$$
\delta=\dfrac{1}{n^{1+o(1)}},
$$
we have that
\begin{equation}
\mathbb{P}\left(\bigcup_{i}\A_i\right)\leq o(1),
\end{equation}
i.e.\ with probability $1-o(1)$, no $\delta_i$ is less than $\frac{1}{n^{5/2+o(1)}\sqrt{p}}$. This in turn implies that,
\begin{equation}
\label{eqmain:min-gap-erg}
\Delta_{\min}\geq \dfrac{1}{n^{5/2+o(1)}\sqrt{p}}.
\end{equation}
with probability $1-o(1)$.
\end{proof}
~\\
We are now equipped with the random matrix theory results, and in the next subsection, we derive a tight upper bound on the double sum $\Sigma$, defined in Eq.~\eqref{eqmain:double-sum-defintion-2}.
\subsection{Upper bound on $\mathbf{\Sigma}$}
As mentioned previously, in order to obtain the quantum mixing time, we first obtain bounds for the double sum $\Sigma$. Recall that we an obtain lower and upper bounds for $\Sigma$ as
$$
\dfrac{1}{\Delta_{\min}}\leq \Sigma \leq \bigOt{\dfrac{n}{\Delta_{\min}}}.
$$
In this subsection, our goal is to obtain an upper bound for $\Sigma$ that is as close as possible to its lower bound. To that end, our strategy would be to make use of the results on the eigenvalue statistics of $\bar{A}_{G(n,p)}$ at macroscopic, mesoscopic and microscopic levels. In particular, in order to evaluate $\lambda_{i+r}-\lambda_i$, for $r<r^\star(i)$, we shall make use of the tail-bounds on consecutive eigenvalue gaps in Eq.~\eqref{eqmain:tail-bound-random-graphs}. On the other hand for $r>r^{\star}(i)$, the eigenvalue rigidity criterion (See Eq.~\eqref{eqmain:eigenvalue-rigidity}) kicks in and we can replace $\lambda_{i+r}-\lambda_i$ with $\gamma_{i+r}-\gamma_i$.
~\\~\\
\textbf{Upper bound on $\mathbf{\Sigma_1}$:~} We first obtain an upper bound on the sum of the inverse of consecutive eigenvalue gaps, i.e.\
\begin{equation}
\Sigma_1=\sum_{i=1}^n \dfrac{1}{\lambda_{i+1}-\lambda_i}.
\end{equation}
In the Supplemental Material of Ref.~\cite{chakraborty2020fast}, we have explicitly derived an upper bound for $\Sigma$. We restate the result here
~\\
\begin{lemma}[\textbf{Upper bound on $\mathbf{\Sigma_1}$ \cite{chakraborty2020fast}}]
\label{lem:sum-bound}
\begin{equation}
\Sigma_1=\sum_{i=1}^{n-1} \frac{1}{\lambda_{i+1} - \lambda_i} \leq n^{5/2+o(1)} \sqrt{p},
\end{equation}
with probability $1-o(1)$.
\end{lemma}
~\\
The key idea is that we count the number of consecutive eigenvalue gaps ($\delta_i$) lying within an interval of $1/\log n$ times the average gap and find that a high fraction of the $\delta_i$-s lie within this window around the average gap. For a detailed derivation, we refer the readers to Ref.~\cite{chakraborty2020fast}.
~\\
Now we can derive an upper bound on $\Sigma$ by combining mesoscopic and microscopic eigenvalue statistics of $\bar{A}_{G(n,p)}$ at different scales of $r$. In particular, we use the upper bound on $\Sigma_1$ along with the eigenvalue rigidity condition. We state the upper bound on $\Sigma$ that we obtained in Ref.~\cite{chakraborty2020fast}.
~\\
\begin{lemma}[\textbf{Upper bound on $\mathbf{\Sigma}$ \cite{chakraborty2020fast}}]
\label{lem:sum-bound-2}
For $p \geq n^{-1/3}$, the eigenvalues of $\bar{A}_{G(n,p)}$ satisfy 
\begin{equation}
\Sigma=\sum_{i=1}^{n-1}\sum_{r=1}^{n-i}\dfrac{1}{\left|\lambda_{i+r}-\lambda_i\right|} \leq n^{5/2-\frac{\log p}{\log n} + o(1)}\sqrt{p},
\end{equation}
with probability $1-o(1)$.
\end{lemma}
~\\
We provide an an intuition of the proof techniques and we refer the reader to the Supplemental Material of Ref.~\cite{chakraborty2020fast} for details. We first split $\Sigma$ into two different parts:
\begin{align}
\label{eqmain:sum-split}
\Sigma=\sum_{r=1}^{r^\star(i)}\Sigma_r + \sum_{r=r^\star(i)+1}^{n-1}\Sigma_r
\end{align}
For the first sum in the right hand side of Eq.~\eqref{eqmain:sum-split}, we are dealing with small eigenvalue gaps and hence we make use of the microscopic eigenvalue statistics, namely the upper bound on $\Sigma_1$, i.e.\ we replace this sum with with the upper bound $r^\star(i).\Sigma_1$. On the other hand, for the second double sum, eigenvalue rigidity provides kicks in and the gaps between the classical eigenvalue locations $(\gamma_{i+r}-\gamma_i)$, is a better estimate of $\lambda_{i+r}-\lambda_i$ than the tail bounds. In fact, an upper bound is obtained by replacing each eigenvalue gap $\lambda_{i+r}-\lambda_i$ with the lower bound from Eq.~\eqref{eqmain:classical-location-distance}.

Observe that for dense random graphs, the upper bound on $\Sigma$ is quite close to its lower bound of $1/\Delta_{\min}$. Having obtained this bound, we shall now upper bound the quantum mixing time for $G(n,p)$.
\subsection{Mixing of continuous-time quantum walks on $\mathbf{G(n,p)}$}

Here, we shall obtain the (i) limiting distribution of the quantum walk and the time after which the quantum walk converges (in a time-averaged sense) to this distribution, the quantum mixing time.

\begin{figure}[h!]
\centering
\includegraphics[trim=0cm 0cm 0cm 0cm, clip=true, width=0.45\textwidth]{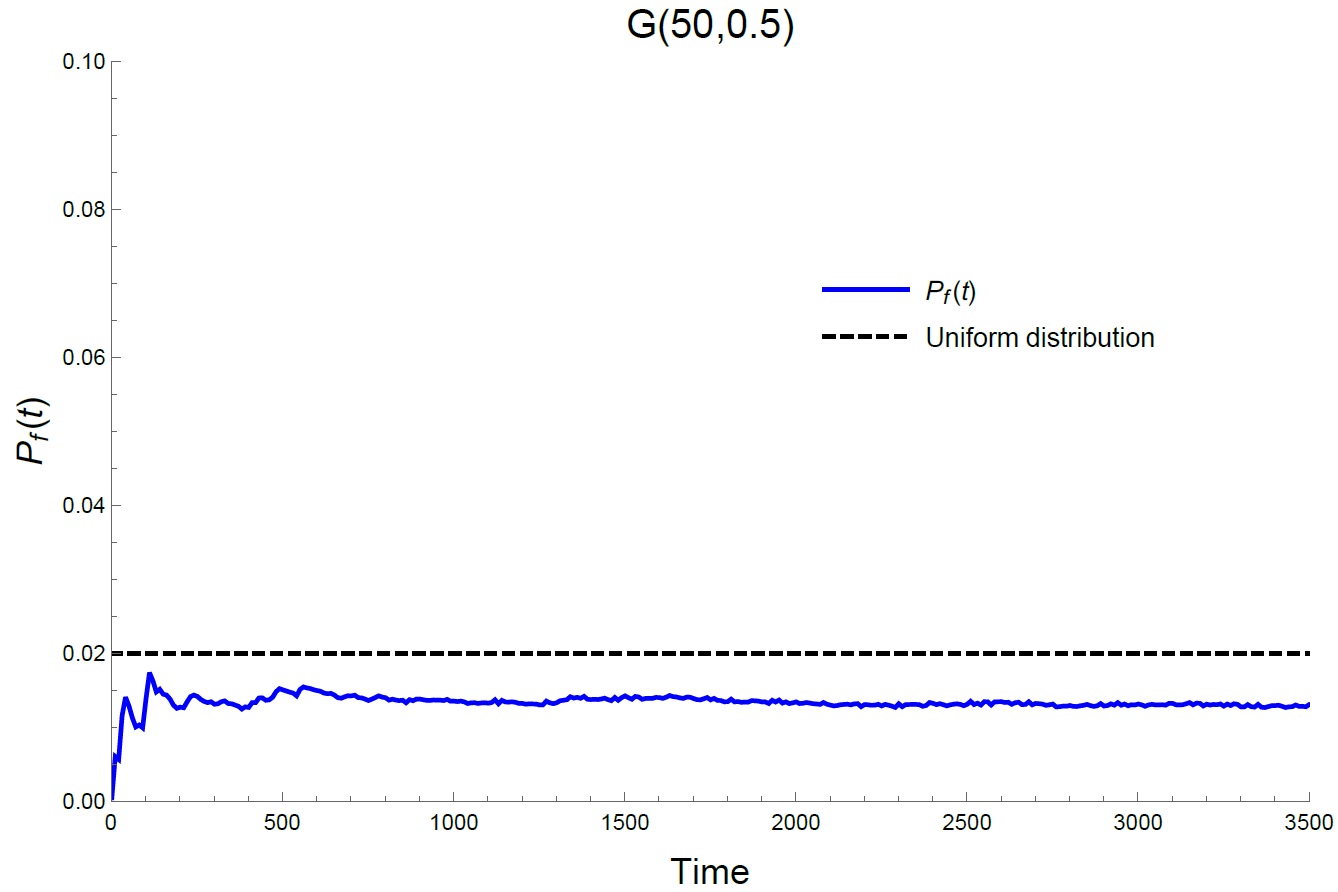}
\caption{\small{The limiting probability distribution is close to the uniform distribution for a quantum walk on $G(n,p)$. The figure shows that the instantaneous time-averaged probability distribution (thick blue line) for a quantum walk on $G(50,0.5)$ remains close to the uniform distribution (horizontal dashed black) after a long enough time.}}
\label{fig:limiting}
\end{figure} 

In order to obtain both these results, we make use of the fact that all the eigenvectors of $\bar{A}_{G(n,p)}$ are completely delocalized. In fact, it was conjectured in Ref.~\cite{dekel2011eigenvectors}, that for dense random graphs, the eigenstates of $\bar{A}_{G(n,p)}$ are completely delocalized. This implies that when any of its eigenvectors $\ket{v_i}$ is expressed in the basis of the nodes of the underlying graph, the absolute value of each entry is at most $n^{-1/2}$ (up to logarithmic factors). Erd\"os et al. \cite{erdHos2013spectral} answered this optimally even for sparse $p$ and the results therein were subsequently extended for any $p$ above the percolation threshold recently by He et al. \cite{he2018local}. This implies that as long as $p\geq \omega(\log(n)/n)$, for all $j\in\{1,\cdots,n\}$
\begin{equation}
\label{eqmain:delocalized-evector}
\|\ket{v_j}\|_\infty\leq n^{-1/2+o(1)},
\end{equation}
with probability $1-o\left(\dfrac{1}{n}\right)$.

To the expression for the limiting probability distribution in Eq.~\eqref{eqmain:prob_infinite}, we substitute the delocalization of eigenvectors from Eq.~\eqref{eqmain:delocalized-evector} to obtain
\begin{align}
P_{f}(T\rightarrow\infty)&=\sum_{i=1}^n|\braket{f}{v_i}\braket{v_i}{\psi_0}|^2\\
                         &\leq \bigOt{1/n}\sum_{i=1}^n|\braket{v_i}{\psi_0}|^2\\
                         &\leq \bigOt{1/n},
\end{align}
independent of $\ket{\psi_0}$, i.e.\ the limiting distribution converges to a (nearly) uniform distribution.

Observe that the upper bound on $\Sigma$ already provides an upper bound on the quantum mixing time. However, we can improve the bound it further if we assume that the quantum walk commences from an \textit{easy to prepare} initial state. By this we mean that the initial state $\ket{\psi_0}$ is a superposition over atmost a $\mathrm{polylog}(n)$ number of nodes. In fact, generally it is assumed that the initial state is localized at some node of the underlying graph, i.e.\ $\ket{\psi_0}=\ket{l}$, which is standard. 

\begin{figure}[h!]
\centering
\includegraphics[trim=0cm 0cm 0cm 0cm, clip=true, width=0.45\textwidth]{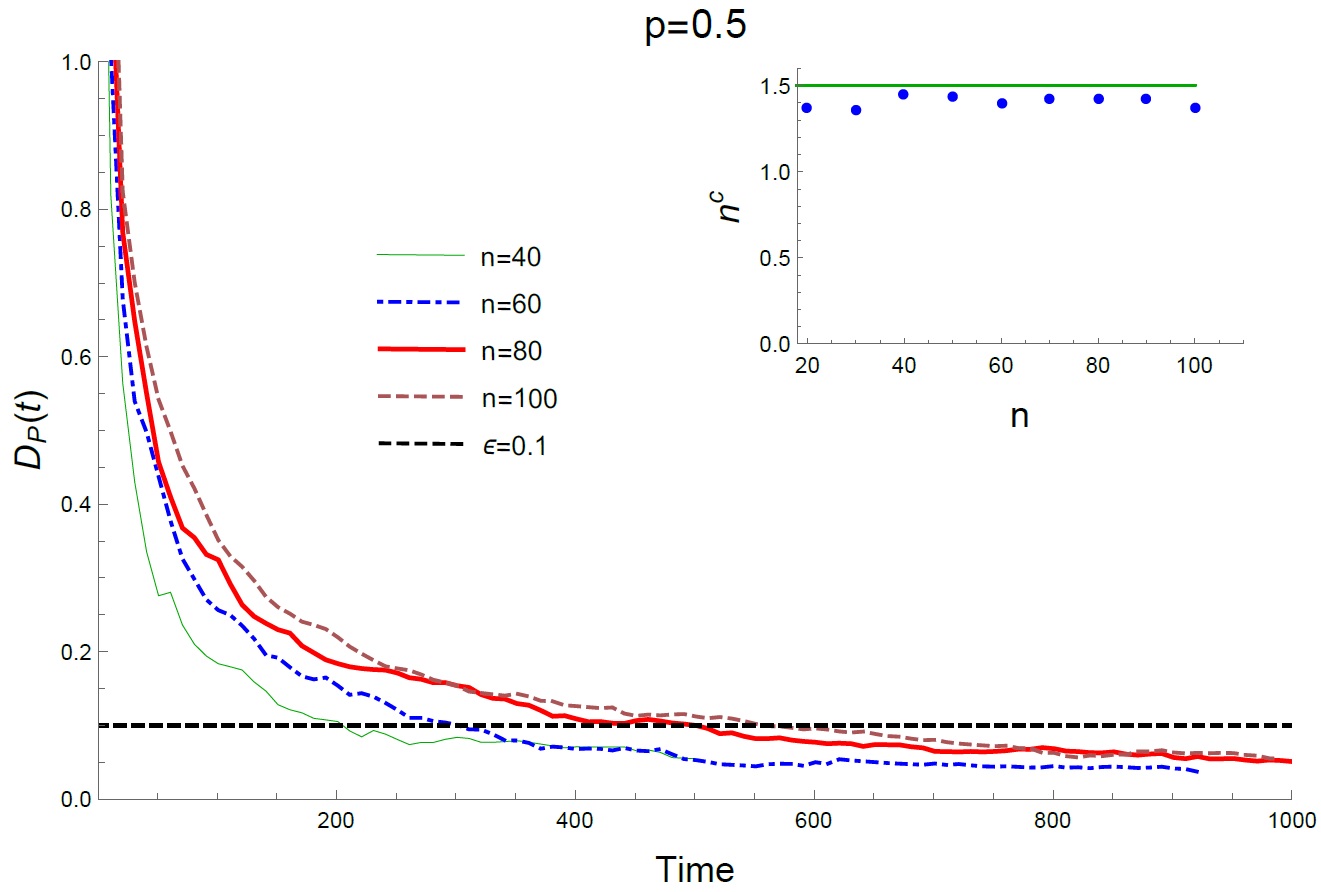}
\caption{\small{Figure shows the time for the instantaneous time-averaged probability distribution at any time $t$, denoted by $P_f(t)$ to be $\epsilon$-close to the limiting probability distribution, $P_f(t\rightarrow\infty)$ for Erd\"os-Renyi random graphs $G(n,p)$. The Y-axis denotes the distance between these two distributions (as measured in one norm), i.e.\ $D_P(t)=\nrm{P_f(t)-P_f(t\rightarrow\infty)}_1$, while the X-axis denotes time. We plot $D_P(t)$ for random graphs of $40$ nodes (dotted green), $60$ nodes (dot-dashed blue), $80$ nodes (solid red) and $100$ nodes (dashed pink), with $p=0.5$. The dotted horizontal line (dashed black) corresponds to $\epsilon=0.1$ which helps indicate the time after which $D_P(t)\leq \epsilon$ for the aforementioned instances. The inset plot shows the exponent $c$ where $n^c$ corresponds to the minimum time after which $D_P(t)\leq 0.1$ (quantum mixing time) for $G(10,0.5),G(20,0.5),\cdots,G(100,0.5)$. The quantum mixing time is thus upper bounded by $n^{3/2}$ which matches with our analytical predictions.}}
\label{fig:convergence-to-limiting}
\end{figure}

If the quantum walk commences from an \textit{easy to prepare} state, 
$$
\ket{\psi_0}=\sum_{k=1}^{q} c_k\ket{k},
$$
where $q$ is in $\mathcal{O}(\mathrm{polylog}(n))$, we can use Eq.~\eqref{eqmain:delocalized-evector} to obtain
\begin{align}
T^{G(n,p)}_{\mathrm{mix}}&=\Oo\left(\dfrac{1}{\epsilon}\sum_{i=1}^{n-1}\sum_{r=1}^{n-i}\dfrac{\left |\braket{v_i}{\psi_0}\right |. \left | \braket{\psi_0}{v_{i+r}}\right |}{\left|\lambda_{i+r}-\lambda_i\right|}\right) \\
						 &= \Oo\left(\dfrac{1}{n^{1-o(1)}\epsilon}\sum_{i=1}^{n-1}\sum_{r=1}^{n-i}\dfrac{\sum_{l=1}^{q}|c_l|.\sum_{m=1}^{q}|c^*_m|}{\left|\lambda_{i+r}-\lambda_i\right|}\right) \\
						 &=	\bigOt{\frac{1}{\epsilon}\  \dfrac{\Sigma}{n}}=\bigOt{\frac{1}{\epsilon}\ n^{3/2 - \frac{\log p}{\log n} + o(1)} \sqrt{p}},
\end{align}
with probability $1 - o(1)$.
~\\~\\
Thus for $n^{-1/3}\leq p \leq 1-n^{-1/3}$, 
\begin{equation}
\label{eqmain:upper-bound-mix-time-sparse-p-high}
T^{G(n,p)}_{\mathrm{mix}}=\widetilde{\mathcal{O}}\left(n^{3/2 - \log p/\log n}\sqrt{p}/\epsilon\right),
\end{equation}
for $p\geq n^{-1/3}$.

Observe that for dense Erd\"os-Renyi random graphs, 
\begin{equation}
\label{eqmain:mixing-time-dense-random-graphs}
T^{G(n,p)}=\bigOt{\dfrac{n^{3/2}}{\epsilon}}.
\end{equation}

Also, as $p$ decreases the upper bound on the mixing time increases. Unfortunately for sparser random graphs, i.e.\ for $p=\log^D(n)/n$, such that $D>8$, we cannot make use of eigenvalue rigidity. However simply using Lemma \ref{lem:sum-bound} along with the observation  that 

$$\sum_{i=1}^{n-r-1} \dfrac{1}{\left|\lambda_{i+r}-\lambda_i\right|} \leq \sum_{i=1}^{n-1} \dfrac{1}{|\lambda_{i+1} - \lambda_i|},$$ 
for $2 \leq r \leq n-1$, gives us a weaker upper bound for the quantum mixing time in such regimes of sparsity. We obtain that
\begin{equation}
\label{eqmain:mixing-time-bound-sparser-random-graphs}
T^{G(n,p)}_{\mathrm{mix}}=\mathcal{O}\left(\dfrac{n^{5/2+o(1)}\sqrt{p}}{\epsilon}\right).
\end{equation}
In fact, the breakdown of rigidity estimates in \cite{erdHos2013spectral} is not an artifact of the proof.  For extremely sparse graphs, the optimal rigidity estimates that hold in dense graphs are known to break down \cite{huang2017transition}.

  Note that there exist weaker forms of rigidity of sparse graphs when $p \leq n^{-1/3}$ which may lead to modest improvements of the exponent of $n$ in the mixing time. However, we have not expended too much effort optimizing the exponent as we are fundamentally limited by the smallest gap, $\Delta_{\min}$ (see lower bound of Eq.~\eqref{eqmain:min-gap-erg}) for which the bounds in \cite{lopatto2019tail} are still quite far from the conjectured behaviour. Obtaining the conjectured smallest gap behaviour appears to be a difficult problem in random matrix theory.
  
  Finally, we numerically verify the analytical results obtained. Fig.~\ref{fig:limiting} shows that for $G(50,0.5)$, the instantaneous time-averaged probability distribution ($P_f(t)$) converges to a distribution that is close to the uniform distribution (horizontal dashed black line). While in Fig.~\ref{fig:convergence-to-limiting} we plot $D_P(t)=\nrm{P_f(t)-P_f(t\rightarrow\infty)}_1$ with time and the inset plot depicts the exponent for the quantum mixing time ($D_P(t)\leq \epsilon$) for random graphs of various sizes and $p=0.5$. The numerical results conform with the analytically obtained upper bound for the quantum mixing time in Eq.~\eqref{eqmain:mixing-time-dense-random-graphs}.
\subsection{Mixing time for continuous-time quantum walks on any ergodic, reversible Markov chain}
\label{subsec:mixing-time-reversible-markov-chain}
Our results thus far have provided an upper bound on the quantum mixing time for almost all simple unweighted graphs. Now we address the quantum mixing time for any ergodic, reversible Markov chain $P$. Any symmetric matrix that captures the local connectivity of $P$ can be a used as a Hamiltonian for performing a CTQW on $P$. As $P$ need not be symmetric in general, one cannot consider a quantum walk on $P$ directly. Given any such Markov chain $P$, one can define the Hamiltonian $H=i[V^\dag S V,\Pi_0]$ as stated in Sec.~\ref{sec-main:hamiltonian-any-markov-chain} (for $s=0$). In this section we consider the limiting distribution of a continuous-time quantum walk under $H$, on the edges of $P$. 


Here, we shall explore whether any generic speedup is obtained for the $QLSamp$ problem. Note that the time evolution of some initial state $\ket{\psi(0),0}$, under the action of $H$ is given by
\begin{align}
\label{eqmain:time-evolution-search-ham}
\ket{\psi(t)}=\braket{v_n}{\psi_0}\ket{v_n,0}+\sum_{j=1,\sigma=\pm}^{n-1}e^{-itE^\sigma_j}\dfrac{\braket{v_j}{\psi_0}}{\sqrt{2}}\ket{\Psi^\sigma_j}.
\end{align}
The limiting probability distribution (note that now we are projecting on obtaining $\ket{0}$ in the second register) is given by

$$P_{f}(T)=\frac{1}{T}\int_{0}^{T} dt\ |\braket{f,0|e^{-iHt}}{\psi_0,0}|^2.$$


\begin{figure}[h!]
\centering
\includegraphics[scale=0.4]{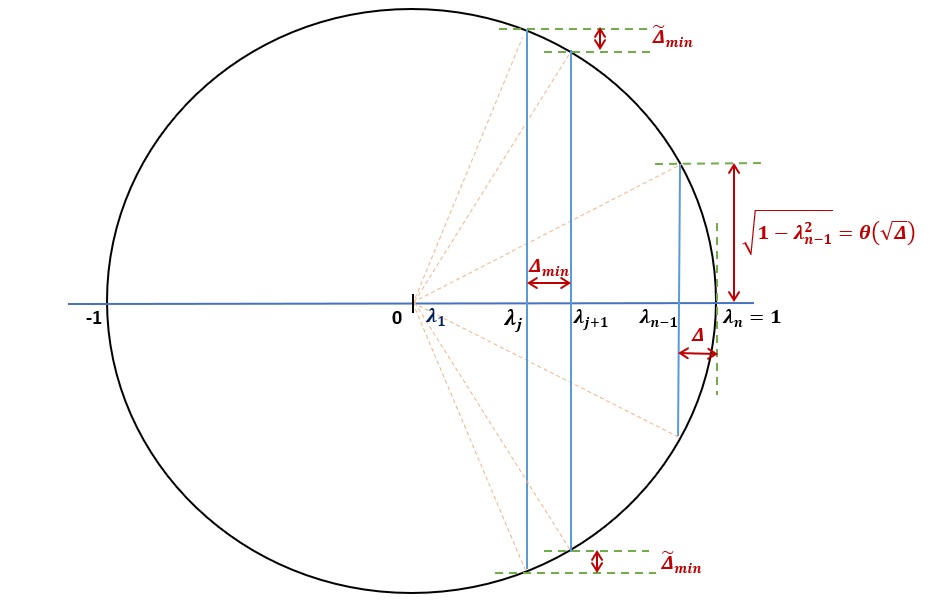}

\caption{\footnotesize{Comparison of the gaps between eigenvalues of a Markov chain $P$ and the corresponding Hamiltonian $H$ defined in Sec.~\ref{sec-main:hamiltonian-any-markov-chain}. The eigenvalues of $P$ lie between $0$ and $1$. Any such eigenvalue $\lambda$ of $P$ is mapped to the eigenvalue pair $\pm \sqrt{1-\lambda^2}$ in the relevant subspace of $H$. As a result the spectral gap, $\Delta$, of $P$ is mapped to $\Theta(\sqrt{\Delta})$ for $H$. However, this is not the case for all eigenvalue gaps. In fact, the minimum over all eigenvalue gaps of $P$, $\Delta_{\min}$ is mapped to $\widetilde{\Delta}_{\min}$, such that $\widetilde{\Delta}_{\min}>\Delta_{\min}$ if $\Delta_{\min}$ appears between two eigenvalues that are close to $\lambda_{n-1}$. On the other hand, $\widetilde{\Delta}_{\min}>\Delta_{\min}$ if it appears between two eigenvalues that are close to $\lambda_2$. This has been elucidated in Sec.~\ref{subsec:mixing-time-reversible-markov-chain}}.}
\label{figmain:ham-spec}
\end{figure}

This implies,
\begin{equation}
\label{eqmain:limiting-somma-ortiz}
P_{f}(T\rightarrow\infty)=\frac{1}{2}\sum_{\lambda_i=\lambda_l}\braket{v_l}{f}\braket{f}{v_i}\braket{v_i}{\psi_0}\braket{\psi_0}{v_l}.
\end{equation}

Also, the upper bound on the quantum mixing time is given by,
\begin{align}
\label{eqmain:upper-bound-simplified-somma-ortiz}
T^{P}_{\mathrm{mix}}&= \Oo\left(\dfrac{1}{\epsilon}\sum_{i\neq l} \dfrac{\left |\braket{E_i}{\psi_0}\right | .\left | \braket{\psi_0}{E_l}\right |}{\left | E_i-E_l\right |}\right),
\end{align}
where recall from Eq.~\eqref{eqmain:eigenpair-somma-ortiz} in Sec.~\ref{subsec:spectrum-somma-ortiz} that $E_j=\sqrt{1-\lambda_j^2}$. 

Now the generic upper bound on the quantum mixing time $T^{P}_{\mathrm{mix}}$ is upper bounded by the double sum $\Sigma$ and as such 
\begin{equation}
\label{eq:upper-bound-mixing-somma-ortiz}
T^{P}_{\mathrm{mix}}\leq \Sigma=\sum_{i\neq l} \dfrac{1}{\left | E_i-E_l\right |} \leq \bigOt{\dfrac{n}{\widetilde{\Delta}_{\min}}},
\end{equation}
where $\widetilde{\Delta}_{\min}$ is the minimum eigenvalue gap of the Hamiltonian $H$. We now need bound $\widetilde{\Delta}_{\min}$ in terms of the minimum eigenvalue gap of $P$, $\Delta_{\min}$. To that end we have the following lemma: 
~\\
\begin{lemma}
\label{lem-main:min-evalue-gap-somma-ortiz}
Suppose $P$ is an ergodic, reversible Markov chain with eigenvalues $\lambda_n=1>\lambda_{n-1}\geq \cdots \lambda_1\geq 0$. Suppose $\Delta$ is the spectral gap of $P$ and the minimum of all gaps between distinct eigenvalues of $P$ be $\Delta_{\min}$. Then the minimum eigenvalue gap of the Hamiltonian $H=i[V^\dag S V, \Pi_0]$, $\widetilde{\Delta}_{\min}$ is bounded as 
$$\Theta\left(\lambda_2\Delta_{\min}\right)\leq \widetilde{\Delta}_{\min}\leq \Theta\left(\dfrac{\Delta_{\min}}{\sqrt{\Delta}}\right).$$
\end{lemma}
~\\
\begin{proof}
We know that for $H$, in the relevant subspace, each eigenvalue of $P$, $\lambda_j$, maps to $\pm\sqrt{1-\lambda_j^2}$.
Thus if $\delta_j=|\lambda_{j+1}-\lambda_{j}|$, then we have
\begin{align}
\widetilde{\delta_j}&=\left|\sqrt{1-\lambda_{j+1}^2}-\sqrt{1-\lambda_j^2}\right|\\
                    &=\left|\sqrt{1-\lambda_{j+1}^2}-\sqrt{1-(\lambda_{j+1}-\delta_j)^2}\right|\\
\label{eqmain:proof-mixing-somma-ortiz-pre-taylor}                    
                    &=\sqrt{1-\lambda_{j+1}^2}\left|1-\sqrt{1+\dfrac{2\delta_j\lambda_{j+1}}{1-\lambda_{j+1}^2}-\dfrac{\delta_j^2}{1-\lambda_{j+1}^2}}\right|.
\end{align}
We are concerned with the minimum eigenvalue gap $\widetilde{\Delta}_{\min}$. Without loss of generality, we assume that $P$ has a simple spectrum (consequently, so does $H$) and for some $1\leq j \leq n-1$, the eigenvalue gap is minimum for two consecutive distinct eigenvalues $\lambda_j$ and $\lambda_{j+1}$. That is, for some value of $j$, $\delta_j=\Delta_{\min}$ and henceforth we consider that value of $j$. Observe that in such a case, the second term inside the square-root
\begin{equation}
\dfrac{2\delta_j\lambda_{j+1}}{1-\lambda_{j+1}^2}=\dfrac{2\Delta_{\min}\lambda_{j+1}}{1-\lambda_{j+1}^2}< \dfrac{2\Delta_{\min}}{\Delta}<1.
\end{equation}
So expanding Eq.~\eqref{eqmain:proof-mixing-somma-ortiz-pre-taylor} according to Taylor series, we have
\begin{align}
\widetilde{\Delta}_{\min}&=\dfrac{2\Delta_{\min}\lambda_{j+1}}{\sqrt{1-\lambda^2_{j+1}}}+\Theta\left(\dfrac{\Delta^2_{\min}}{\sqrt{1-\lambda^2_{j+1}}}\right)\\
                         &= \Theta\left(\dfrac{\Delta_{\min}\lambda_{j+1}}{\sqrt{1-\lambda^2_{j+1}}}\right).
\end{align}
This expression implies that the minimum eigenvalue gap of $P$ is mapped to the minimum eigenvalue gap of $H$ multiplied by the ratio of the corresponding eigenvalues of $P$ and $H$. The upper and lower bounds follow from observing that for all $1\leq j \leq n-1$, $\sqrt{1-\lambda^2_{j+1}}\leq \Theta(\sqrt{\Delta})$ and $\lambda_{j+1}=\Omega\left(\lambda_2\right)$, respectively. 
\end{proof}
~\\
So from Lemma \ref{lem-main:min-evalue-gap-somma-ortiz}, we have that for any ergodic, reversible Markov chain $P$
\begin{equation}
\label{eqmain:quantum-mixing-time-ub-hamP}
T^{P}_{\mathrm{mix}}= \Oo\left(\dfrac{1}{\epsilon} \dfrac{n}{\lambda_2\Delta_{\min}}\right).
\end{equation} 

Let us now consider that $P$ is a symmetric, i.e.\ $P=P^T$. Then the underlying quantum walk can also be performed on $P$ itself. Assuming that the eigenvalues of $P$ are ordered, for a continuous-time quantum walk on $H$, from Eq.~\eqref{eq:upper-bound-mixing-somma-ortiz} and Lemma \ref{lem-main:min-evalue-gap-somma-ortiz}, we observe that the upper bound for the quantum mixing time may be faster or slower than a quantum walk performed on $P$ depending on where the minimum eigenvalue gap appears (See Fig.~\ref{figmain:ham-spec} for a pictorial representation). 

If $\Delta_{\min}$ happens to be between two eigenvalues that are close to $\lambda_{n-1}$, $\tilde{\Delta}_{\min}\approx\Delta_{\min}/\sqrt{\Delta}$ and hence the upper bound on the quantum mixing time is in $\bigOt{n\sqrt{\Delta}/\Delta_{\min}}$ which is faster than the bound in Eq.~\eqref{eqmain:mixing-time-generic-ub}. On the other hand, if $\Delta_{\min}$ is in the vicinity of $\lambda_2$, the upper bound on the quantum mixing time is given by Eq.~\eqref{eqmain:quantum-mixing-time-ub-hamP}. 

For generic ergodic, reversible Markov chains however, this comparison is inapplicable as $P$ may not be symmetric and cannot be used as a CTQW Hamiltonian.    

This is in contrast to the $QSSamp$ problem, where using $H$ offers a generic quadratic speedup over using $P$ as the Hamiltonian in Algorithm \ref{algo1}. This shows a fundamental difference between the two different notions of mixing for quantum algorithms as elucidated by $QSSamp$ and $QLSamp$ problems.
\section{Discussion}
\label{sec-main:discussion}
In this article we have discussed the two notions of quantum  mixing and designed analog quantum algorithms to tackle these problems. First, using Hamiltonian evolution and von Neumann measurements, we have presented an analog quantum algorithm that, given an ergodic, reversible Markov chain outputs a coherent encoding of its stationary state. The running time of our algorithm matches that of its discrete-time counterparts. Secondly, we have also discussed the problem of sampling from the limiting distribution of a (time-averaged) continuous-time quantum walk. We have offered an intuitive explanation of the tools used in Ref.~\cite{chakraborty2020fast} to derive upper bounds on the mixing time for random graphs. We have also backed up the analytical results therein with numerical simulations and extended the time-averaged notion of mixing to any ergodic, reversible Markov chain.

Our results could pave the way for further research. For example, quantum state-generation using von Neumann measurements can be used to develop novel analog quantum algorithms. Note that our methods could be used to obtain other analog quantum algorithms for solving the $QSSamp$ problem. One could reverse the spatial search algorithm by Childs and Goldstone \cite{childs2004spatial, chakraborty2020optimality} and use von Neumann measurements to prepare a coherent encoding of the highest eigenstate of underlying Hamiltonian. In the case of state-transitive graphs, this will allow for uniform sampling. 

It would be interesting to explore whether using our framework, one can construct an analog quantum algorithm to fast-forward the dynamics of any ergodic, reversible Markov chain much like the results of Apers and Sarlette in discrete-time \cite{simon2018quantum}. The challenge is that most of the underlying techniques that enable this, such as the recently developed techniques in the context of quantum simulation \cite{low2016hamiltonian, chakraborty2018power, gilyen2018quantum}, are absent in continuous-time. However, the fact that the Hamiltonian defining our continuous-time quantum walk can be efficiently simulated using query access to the unitary defining the discrete-time quantum walk of Ref.~\cite{krovi2016quantum} might offer useful insights towards designing such algorithms. 

Our algorithm can also be used to prepare stationary states of \textit{slowly-evolving} Markov chains, i.e.\ given a sequence of Markov chains $\{P_1,\cdots,P_n\}$, such that there is a significant overlap between the stationary distributions of any two consecutive Markov chains, meaning  $|\braket{\pi_{j+1}}{\pi_j}|$ is large \cite{aharonov2003adiabatic, wocjan2008speedup, orsucci2018optimal}. Given that one can prepare $\ket{\pi_1}$ efficiently, the task is to prepare $\ket{\pi_n}$. Such situations arise in a host of approximation algorithms for counting as has been pointed out in Ref.~\cite{aharonov2003adiabatic}. Our algorithm will provide a quadratic speedup over that of Ref.~\cite{aharonov2003adiabatic} as given any $P_j$, the spectral gap of the Hamiltonian defined in Sec.~\ref{sec-main:hamiltonian-any-markov-chain} is amplified quadratically over the corresponding discriminant matrix, which acts as the Hamiltonian for the approach in \cite{aharonov2003adiabatic}.

For the problem of time-averaged mixing, it would be interesting to explore the possibility of obtaining better bounds on the quantum mixing time for any ergodic, reversible Markov chain. Furthermore, this notion of quantum mixing is closely related to the problem of equilibration of isolated quantum systems, a widely studied problem in quantum statistical mechanics \cite{wilming2018equilibration}. As a result, our results can help obtain better upper bounds for the equilibration times of isolated quantum systems defined by random Hamiltonians.
\begin{acknowledgments}
S.C. acknowledges funding from F.R.S.-FNRS. S.C. and J.R. are supported by the Belgian Fonds de la Recherche Scientifique - FNRS under grants no F.4515.16 (QUICTIME) and R.50.05.18.F (QuantAlgo). K.L. has been partially supported by NSF postdoctoral fellowship DMS-1702533.
\end{acknowledgments}
\bibliographystyle{unsrt}
\bibliography{References}

\begin{thebibliography}{10}

\bibitem{binder1993monte}
Kurt Binder, Dieter Heermann, Lyle Roelofs, A~John Mallinckrodt, and Susan
  McKay.
\newblock Monte carlo simulation in statistical physics.
\newblock {\em Computers in Physics}, 7(2):156--157, 1993.

\bibitem{aarts1985statistical}
Emile~HL Aarts and Peter~JM Van~Laarhoven.
\newblock Statistical cooling: A general approach to combinatorial optimization
  problems.
\newblock {\em Philips J. Res.}, 40(4):193--226, 1985.

\bibitem{page1999pagerank}
Lawrence Page, Sergey Brin, Rajeev Motwani, and Terry Winograd.
\newblock The pagerank citation ranking: Bringing order to the web.
\newblock Technical report, Stanford InfoLab, 1999.

\bibitem{gilks1995markov}
Walter~R Gilks, Sylvia Richardson, and David Spiegelhalter.
\newblock {\em Markov chain Monte Carlo in practice}.
\newblock Chapman and Hall/CRC, 1995.

\bibitem{Note1}
Throughout the article, we use $\protect \mathaccentV {tilde}07E{\protect
  \mathcal {O}}(.)$ and $\protect \mathaccentV {tilde}07E{\Omega }(.)$ to
  suppress polylogarithmic factors, i.e.\ $\protect \mathaccentV
  {tilde}07E{\protect \mathcal {O}}(f(n))=\protect \mathcal {O}(f(n)\protect
  \mathrm {polylog}(f(n)))$ and $\protect \mathaccentV {tilde}07E{\Omega
  }(f(n))=\Omega (f(n)\protect \mathrm {polylog}(f(n)))$.

\bibitem{chakraborty2020fast}
Shantanav Chakraborty, Kyle Luh, and J{\'e}r{\'e}mie Roland.
\newblock How fast do quantum walks mix?
\newblock {\em Physical Review Letters}, 124(5):050501, 2020.

\bibitem{aharonov2003adiabatic}
Dorit Aharonov and Amnon Ta-Shma.
\newblock Adiabatic quantum state generation and statistical zero knowledge.
\newblock In {\em Proceedings of the thirty-fifth annual ACM symposium on
  Theory of computing}, pages 20--29. ACM, 2003.

\bibitem{aho1974design}
Alfred~V Aho and John~E Hopcroft.
\newblock {\em The design and analysis of computer algorithms}.
\newblock Pearson Education India, 1974.

\bibitem{richter2007quantum}
Peter~C Richter.
\newblock Quantum speedup of classical mixing processes.
\newblock {\em Physical Review A}, 76(4):042306, 2007.

\bibitem{wocjan2008speedup}
Pawel Wocjan and Anura Abeyesinghe.
\newblock Speedup via quantum sampling.
\newblock {\em Physical Review A}, 78(4):042336, 2008.

\bibitem{dunjko2015quantum}
Vedran Dunjko and Hans~J Briegel.
\newblock Quantum mixing of markov chains for special distributions.
\newblock {\em New Journal of Physics}, 17(7):073004, 2015.

\bibitem{paparo2014quantum}
Giuseppe~Davide Paparo, Vedran Dunjko, Adi Makmal, Miguel~Angel Martin-Delgado,
  and Hans~J Briegel.
\newblock Quantum speedup for active learning agents.
\newblock {\em Physical Review X}, 4(3):031002, 2014.

\bibitem{orsucci2018optimal}
Davide Orsucci, Hans~J Briegel, and Vedran Dunjko.
\newblock Faster quantum mixing for slowly evolving sequences of markov chains.
\newblock {\em Quantum}, 2:105, 2018.

\bibitem{dunjko2018machine}
Vedran Dunjko and Hans~J Briegel.
\newblock Machine learning \& artificial intelligence in the quantum domain: a
  review of recent progress.
\newblock {\em Reports on Progress in Physics}, 81(7):074001, 2018.

\bibitem{richter2007almost}
Peter~C Richter.
\newblock Almost uniform sampling via quantum walks.
\newblock {\em New Journal of Physics}, 9(3):72, 2007.

\bibitem{szegedy2004quantum}
Mario Szegedy.
\newblock Quantum speed-up of markov chain based algorithms.
\newblock In {\em Proceedings. 45th Annual IEEE Symposium on Foundations of
  Computer Science, 2004.}, pages 32--41. IEEE, 2004.

\bibitem{magniez2011search}
Fr{\'e}d{\'e}ric Magniez, Ashwin Nayak, J{\'e}r{\'e}mie Roland, and Miklos
  Santha.
\newblock Search via quantum walk.
\newblock {\em SIAM Journal on Computing}, 40(1):142--164, 2011.

\bibitem{krovi2016quantum}
Hari Krovi, Fr{\'e}d{\'e}ric Magniez, Maris Ozols, and J{\'e}r{\'e}mie Roland.
\newblock Quantum walks can find a marked element on any graph.
\newblock {\em Algorithmica}, 74(2):851--907, 2016.

\bibitem{ambainis2019quadratic}
Andris Ambainis, Andr{\'a}s Gily{\'e}n, Stacey Jeffery, and Martins Kokainis.
\newblock Quadratic speedup for finding marked vertices by quantum walks.
\newblock {\em In Proceedings of the 52nd Annual ACM SIGACT Symposium on Theory
  of Computing (STOC)}, pages 412--424, 2020.

\bibitem{cleve1998quantum}
Richard Cleve, Artur Ekert, Chiara Macchiavello, and Michele Mosca.
\newblock Quantum algorithms revisited.
\newblock {\em Proceedings of the Royal Society of London. Series A:
  Mathematical, Physical and Engineering Sciences}, 454(1969):339--354, 1998.

\bibitem{brassard2002quantum}
Gilles Brassard, Peter Hoyer, Michele Mosca, and Alain Tapp.
\newblock Quantum amplitude amplification and estimation.
\newblock {\em Contemporary Mathematics}, 305:53--74, 2002.

\bibitem{simon2018quantum}
Simon Apers and Alain Sarlette.
\newblock Quantum fast-forwarding markov chains.
\newblock {\em arXiv preprint arXiv:1804.02321}, 2018.

\bibitem{chakraborty2018finding}
Shantanav Chakraborty, Leonardo Novo, and J{\'e}r{\'e}mie Roland.
\newblock Finding a marked node on any graph by continuous time quantum walk.
\newblock {\em arXiv preprint arXiv:1807.05957}, 2018.

\bibitem{von1955mathematical}
John Von~Neumann.
\newblock {\em Mathematical foundations of quantum mechanics}.
\newblock Princeton university press, 1955.

\bibitem{childs2002quantum}
Andrew~M Childs, Enrico Deotto, Edward Farhi, Jeffrey Goldstone, Sam Gutmann,
  and Andrew~J Landahl.
\newblock Quantum search by measurement.
\newblock {\em Physical Review A}, 66(3):032314, 2002.

\bibitem{somma2010quantum}
RD~Somma and G~Ortiz.
\newblock Quantum approach to classical thermodynamics and optimization.
\newblock In {\em Quantum Quenching, Annealing and Computation}, pages 1--20.
  Springer, 2010.

\bibitem{aharonov2001quantum}
Dorit Aharonov, Andris Ambainis, Julia Kempe, and Umesh Vazirani.
\newblock Quantum walks on graphs.
\newblock In {\em Proceedings of the thirty-third annual ACM symposium on
  Theory of computing}, pages 50--59. ACM, 2001.

\bibitem{ahmadi2003mixing}
Amir Ahmadi, Ryan Belk, Christino Tamon, and Carolyn Wendler.
\newblock On mixing in continuous-time quantum walks on some circulant graphs.
\newblock {\em Quantum Information \& Computation}, 3(6):611--618, 2003.

\bibitem{kendon2003decoherence}
Viv Kendon and Ben Tregenna.
\newblock Decoherence can be useful in quantum walks.
\newblock {\em Physical Review A}, 67(4):042315, 2003.

\bibitem{fedichkin2005mixing}
Leonid Fedichkin, Dmitry Solenov, and Christino Tamon.
\newblock Mixing and decoherence in continuous-time quantum walks on cycles.
\newblock {\em arXiv preprint quant-ph/0509163}, 2005.

\bibitem{marquezino2008mixing}
FL~Marquezino, R~Portugal, G~Abal, and R~Donangelo.
\newblock Mixing times in quantum walks on the hypercube.
\newblock {\em Physical Review A}, 77(4):042312, 2008.

\bibitem{marquezino2010mixing}
Franklin~L Marquezino, Renato Portugal, and Gonzalo Abal.
\newblock Mixing times in quantum walks on two-dimensional grids.
\newblock {\em Physical Review A}, 82(4):042341, 2010.

\bibitem{kieferova2012quantum}
Maria Kieferova and Daniel Nagaj.
\newblock Quantum walks on necklaces and mixing.
\newblock {\em International Journal of Quantum Information}, 10(02):1250025,
  2012.

\bibitem{norris1998markov}
James~R Norris.
\newblock {\em Markov chains}.
\newblock Number~2. Cambridge university press, 1998.

\bibitem{Note2}
Note that throughout this paper, we shall be dealing with values of $s\in
  [0,1)$ and so the properties of $D(P(1))$ are not relevant here.

\bibitem{aldous1982some}
David~J Aldous.
\newblock Some inequalities for reversible markov chains.
\newblock {\em Journal of the London Mathematical Society}, 2(3):564--576,
  1982.

\bibitem{krovi2010adiabatic}
Hari Krovi, Maris Ozols, and J{\'e}r{\'e}mie Roland.
\newblock Adiabatic condition and the quantum hitting time of markov chains.
\newblock {\em Physical Review A}, 82(2):022333, 2010.

\bibitem{childs2004spatial}
Andrew~M Childs and Jeffrey Goldstone.
\newblock Spatial search by quantum walk.
\newblock {\em Physical Review A}, 70(2):022314, 2004.

\bibitem{janmark2014global}
Jonatan Janmark, David~A Meyer, and Thomas~G Wong.
\newblock Global symmetry is unnecessary for fast quantum search.
\newblock {\em Physical Review Letters}, 112(21):210502, 2014.

\bibitem{foulger2014quantum}
Iain Foulger, Sven Gnutzmann, and Gregor Tanner.
\newblock Quantum search on graphene lattices.
\newblock {\em Physical review letters}, 112(7):070504, 2014.

\bibitem{childs2014spatial-crystal}
Andrew~M. Childs and Yimin Ge.
\newblock Spatial search by continuous-time quantum walks on crystal lattices.
\newblock {\em Phys. Rev. A}, 89:052337, May 2014.

\bibitem{meyer2015connectivity}
David~A Meyer and Thomas~G Wong.
\newblock Connectivity is a poor indicator of fast quantum search.
\newblock {\em Physical review letters}, 114(11):110503, 2015.

\bibitem{novo2015systematic}
Leonardo Novo, Shantanav Chakraborty, Masoud Mohseni, Hartmut Neven, and Yasser
  Omar.
\newblock Systematic dimensionality reduction for quantum walks: Optimal
  spatial search and transport on non-regular graphs.
\newblock {\em Scientific reports}, 5:13304, 2015.

\bibitem{philipp2016continuous}
Pascal Philipp, Lu{\'\i}s Tarrataca, and Stefan Boettcher.
\newblock Continuous-time quantum search on balanced trees.
\newblock {\em Physical Review A}, 93(3):032305, 2016.

\bibitem{wong2016quantum}
Thomas~G Wong.
\newblock Quantum walk search on johnson graphs.
\newblock {\em Journal of Physics A: Mathematical and Theoretical},
  49(19):195303, 2016.

\bibitem{chakraborty2016spatial}
Shantanav Chakraborty, Leonardo Novo, Andris Ambainis, and Yasser Omar.
\newblock Spatial search by quantum walk is optimal for almost all graphs.
\newblock {\em Physical review letters}, 116(10):100501, 2016.

\bibitem{chakraborty2017optimal}
Shantanav Chakraborty, Leonardo Novo, Serena Di~Giorgio, and Yasser Omar.
\newblock Optimal quantum spatial search on random temporal networks.
\newblock {\em Physical review letters}, 119(22):220503, 2017.

\bibitem{chakraborty2020optimality}
Shantanav Chakraborty, Leonardo Novo, and J{\'e}r{\'e}mie Roland.
\newblock On the optimality of spatial search by continuous-time quantum walk.
\newblock {\em arXiv preprint arXiv:2004.12686}, 2020.

\bibitem{yoder2014fixed}
Theodore~J Yoder, Guang~Hao Low, and Isaac~L Chuang.
\newblock Fixed-point quantum search with an optimal number of queries.
\newblock {\em Physical review letters}, 113(21):210501, 2014.

\bibitem{graph_enumeration}
Frank Harary and Edgar~M Palmer.
\newblock {\em Graphical enumeration}.
\newblock Elsevier, 2014.

\bibitem{ER59}
P.~Erd{\H{o}}s and A.~R{\'{e}}nyi.
\newblock On random graphs. {I}.
\newblock {\em Publ. Math. Debrecen}, 6:290--297, 1959.

\bibitem{ER60}
P.~Erd\H{o}s and A.~R\'{e}nyi.
\newblock {On the evolution of random graphs}.
\newblock {\em Publications of the Mathematical Institute of the Hungarian
  Academy of Sciences}, 5:17--61, 1960.

\bibitem{bollobas_book}
B{\'e}la Bollob{\'a}s.
\newblock {\em Random graphs}.
\newblock Springer, 1998.

\bibitem{furedi1981eigenvalues}
Zolt{\'a}n F{\"u}redi and J{\'a}nos Koml{\'o}s.
\newblock The eigenvalues of random symmetric matrices.
\newblock {\em Combinatorica}, 1(3):233--241, 1981.

\bibitem{erdHos2013spectral}
L{\'a}szl{\'o} Erd{\"{o}}s, Antti Knowles, Horng-Tzer Yau, Jun Yin, et~al.
\newblock Spectral statistics of erd{\H{o}}s--r{\'e}nyi graphs i: local
  semicircle law.
\newblock {\em The Annals of Probability}, 41(3B):2279--2375, 2013.

\bibitem{vu2007}
Van~H. Vu.
\newblock Spectral norm of random matrices.
\newblock {\em Combinatorica}, 27(6):721--736, 2007.

\bibitem{erdHos2012rigidity}
L{\'a}szl{\'o} Erd{\H{o}}s, Horng-Tzer Yau, and Jun Yin.
\newblock Rigidity of eigenvalues of generalized wigner matrices.
\newblock {\em Advances in Mathematics}, 229(3):1435--1515, 2012.

\bibitem{tao2014random}
Terence Tao and Van Vu.
\newblock Random matrices have simple spectrum.
\newblock {\em arXiv preprint arXiv:1412.1438}, 2014.

\bibitem{babai1982isomorphism}
L{\'a}szl{\'o} Babai, D~Yu Grigoryev, and David~M Mount.
\newblock Isomorphism of graphs with bounded eigenvalue multiplicity.
\newblock In {\em Proceedings of the fourteenth annual ACM symposium on Theory
  of computing}, pages 310--324. ACM, 1982.

\bibitem{luh2018sparse}
Kyle Luh and Van Vu.
\newblock Sparse random matrices have simple spectrum.
\newblock {\em arXiv preprint arXiv:1802.03662}, 2018.

\bibitem{nguyen2015random}
Hoi Nguyen, Terence Tao, and Van Vu.
\newblock Random matrices: tail bounds for gaps between eigenvalues.
\newblock {\em Probability Theory and Related Fields}, pages 1--40, 2015.

\bibitem{lopatto2019tail}
Patrick Lopatto and Kyle Luh.
\newblock Tail bounds for gaps between eigenvalues of sparse random matrices.
\newblock {\em arXiv preprint arXiv:1901.05948}, 2019.

\bibitem{dekel2011eigenvectors}
Yael Dekel, James~R Lee, and Nathan Linial.
\newblock Eigenvectors of random graphs: Nodal domains.
\newblock {\em Random Structures \& Algorithms}, 39(1):39--58, 2011.

\bibitem{he2018local}
Yukun He, Antti Knowles, and Matteo Marcozzi.
\newblock Local law and complete eigenvector delocalization for supercritical
  erd\"{o}s-r{\'e}nyi graphs.
\newblock {\em arXiv preprint arXiv:1808.09437}, 2018.

\bibitem{huang2017transition}
Jiaoyang Huang, Benjamin Landon, and Horng-Tzer Yau.
\newblock Transition from tracy-widom to gaussian fluctuations of extremal
  eigenvalues of sparse erd$\backslash$h $\{$o$\}$ sr$\backslash$'enyi graphs.
\newblock {\em arXiv preprint arXiv:1712.03936}, 2017.

\bibitem{low2016hamiltonian}
Guang~Hao Low and Isaac~L Chuang.
\newblock Hamiltonian simulation by qubitization.
\newblock {\em Quantum}, 3:163, 2019.

\bibitem{chakraborty2018power}
Shantanav Chakraborty, Andr{\'a}s Gily{\'e}n, and Stacey Jeffery.
\newblock The power of block-encoded matrix powers: improved regression
  techniques via faster hamiltonian simulation.
\newblock {\em Proceedings of the 46th International Colloquium on Automata,
  Languages and Programming (ICALP)}, pages 33:1--33:14, 2019.

\bibitem{gilyen2018quantum}
Andr{\'a}s Gily{\'e}n, Yuan Su, Guang~Hao Low, and Nathan Wiebe.
\newblock Quantum singular value transformation and beyond: exponential
  improvements for quantum matrix arithmetics.
\newblock {\em Proceedings of the 51st Annual ACM SIGACT Symposium on Theory of
  Computing}, pages 193--204, 2019.

\bibitem{wilming2018equilibration}
Henrik Wilming, Thiago~R de~Oliveira, Anthony~J Short, and Jens Eisert.
\newblock Equilibration times in closed quantum many-body systems.
\newblock In {\em Thermodynamics in the Quantum Regime}, pages 435--455.
  Springer, 2018.

\end{thebibliography}
\end{document}